\documentclass[a4paper,onecolumn,11pt, accepted=2026-11-10]{quantumarticle}
\pdfoutput=1
\usepackage[utf8]{inputenc}
\usepackage[english]{babel}
\usepackage[T1]{fontenc}
\usepackage{amsmath}
\usepackage{amsthm}
\usepackage{amssymb}
\usepackage[numbers]{natbib}
\usepackage{hyperref}
\hypersetup{
    colorlinks,
    linkcolor={blue!50!black},
    citecolor={blue!50!black},
    urlcolor={blue!80!black}
}

\usepackage{graphicx}%
\usepackage{bm}%
\usepackage[table]{xcolor}
\usepackage{enumitem}
\usepackage{mathtools}
\usepackage{tikz}
\usepackage{comment}
\usepackage{multirow}
\usepackage{adjustbox}
\usepackage{makecell}
\usepackage{soul}

\usepackage{cancel}

\usepackage[capitalize,noabbrev]{cleveref}
\usepackage{physics}
\usepackage{nicefrac}
\usepackage{textcomp}

\usepackage{tensor}
\usepackage{bbm}
\usepackage{dsfont}
\usepackage[normalem]{ulem}

\renewcommand{\S}{\mathcal{S}}
\newcommand{\iu}{\mathrm{i}}   
\newcommand{\gs}{\ket{\psi_{\text{GS}}}}
\renewcommand{\H}{\mathbb{H}}
\newcommand{\E}{\mathbb{E}}
\setlist{nolistsep}
\definecolor{airforceblue}{rgb}{0.36, 0.54, 0.66}

\makeatletter
\newcommand{\equalcontrib}{\listcsgadd{author\@currentauthornumber affiliations}{*}}
\makeatother

\newtheorem{theorem}{Theorem}
\newtheorem{lemma}{Lemma}

\newtheorem*{corollary*}{Corollary}

\newtheorem*{theorem*}{Theorem}
\newtheorem*{lemma*}{Lemma}
\newtheorem*{fact*}{Fact}
\newtheorem*{remark*}{Remark}

\makeatletter
        \def\CT@@do@color{%
        \global\let\CT@do@color\relax
        \@tempdima\wd\z@
        \advance\@tempdima\@tempdimb
        \advance\@tempdima\@tempdimc
        \advance\@tempdimb\tabcolsep
        \advance\@tempdimc\tabcolsep
        \advance\@tempdima2\tabcolsep
        \kern-\@tempdimb
        \leaders\vrule
        \hskip\@tempdima\@plus  1fill
        \kern0.1pt
        \kern-\@tempdimc
        \hskip-\wd\z@ \@plus -1fill }
\makeatother

\begin{document}

\title{Symmetries, correlation functions, and entanglement of general quantum Motzkin spin-chains}

\author{Varun Menon}
\affiliation{Department of Physics, Harvard University, 17 Oxford Street, Cambridge, MA 02138, USA}
\equalcontrib

\author{Andi Gu}
\affiliation{Department of Physics, Harvard University, 17 Oxford Street, Cambridge, MA 02138, USA}
\equalcontrib

\author{Ramis Movassagh}
\affiliation{Google Quantum AI, Venice, CA, 90291, USA}

\begin{abstract}
Motzkin spin-chains, which include `colorless' (integer spin $s=1$) and `colorful' ($s \geq 2$) variants, are one-dimensional (1D) local integer spin models notable for their lack of a conformal field theory (CFT) description of their low-energy physics, despite being gapless. The colorful variants are particularly unusual, as they exhibit power-law violation of the area-law of entanglement entropy (as $\sqrt{n}$ in system size $n$), rather than a logarithmic violation as seen in a CFT. In this work, we analytically discover several unique properties of these models, potentially suggesting a new universality class for their low-energy physics. We identify a complex structure of symmetries and unexpected scaling behavior in spin-spin correlations, which deviate from known 1D universality classes. Specifically, the $s=1$ chain exhibits $U(1)$ spontaneous symmetry breaking and ferromagnetic order. Meanwhile, the $s \geq 2$ chains do not appear to spontaneously break any symmetries, but display quasi-long-range algebraic order with power-law decaying correlations, inconsistent with standard Berezinskii-Kosterlitz-Thouless (BKT) critical exponents. We also derive asymptotic scaling expressions, with explicit additive constants, for entanglement measures in both colorless and colorful chains, generalizing previous results of Movassagh [J. Math Phys. (2017)], while providing benchmarks for potential quantum simulation experiments. The $\sqrt{n}$ entanglement scaling implies that any matrix-product-state description of these ground states requires bond dimension $e^{\Omega(\sqrt{n})}$; this superpolynomial cost of the standard 1D tensor-network toolkit, along with the analytical tractability of their ground state properties, positions Motzkin spin chains as intriguing testbeds for quantum simulation of many-body physics. 

\end{abstract}

\maketitle
{\renewcommand{\thefootnote}{*}\footnotetext{Authors contributed equally}}

The study of entanglement properties and correlation functions of quantum many-body systems has been a central focus of research at the intersection of condensed matter physics, quantum information theory, and quantum computing in recent years. Models of interacting quantum systems that possess unique highly entangled ground states for which exact properties can be calculated analytically are of particular interest, as they serve as tractable systems to understand universal properties of quantum critical points.

Motzkin spin-chains constitute a class of one-dimensional integer spin-$s$ chains with nearest-neighbor ($2$-local) interactions that is translation-invariant except at the boundaries of the chains. They have emerged as examples of `physical' quantum systems (i.e., systems that are local, translation invariant almost everywhere, and have low spin degrees of freedom) whose unique ground states possess anomalously large entanglement entropy that scales as $\sqrt{n}$ with the system size $n$~\cite{movassagh2016supercritical}.
This scaling violates the area-law for entanglement entropy in 1D by a superlogarithmic factor, providing the first counterexample to the folk belief that ground states of physical, local Hamiltonians in 1D can at most violate the area-law by a logarithmic correction, as in critical systems described by a conformal field theory \cite{calabrese2004entanglement}. Follow up work has shown that a deformation of this class of models gives maximum violation of the area law where the entanglement entropy scales linearly with the system size~\cite{zhang2017novel}, and in this case the gap vanishes exponentially with the square of the system's size~\cite{levine2017gap}.

Previous work has established several key properties of these systems. For the colorless ($s=1$) model, it has been shown that the entanglement entropy between a block of $L$ contiguous spins in the bulk of the chain and its complement asymptotically scales logarithmically with system size \cite{movassagh2017entanglement}. Additionally, spin-spin correlation functions in the $Z$-basis vanish in the bulk of the chain as $O(n^{-1})$ in the thermodynamic ($n\to \infty$) limit. In contrast, for the colorful ($s\geq 2$) models, the half-chain entanglement entropy scales asymptotically as $\sim \sqrt{n}$ for a chain of length $n$~\cite{movassagh2016supercritical}, exhibiting dramatically different behavior than the $s=1$ model. However, in \cite{movassagh2016supercritical} it was shown that the gap of Motzkin Hamiltonians provably vanishes as $O(n^{-z})$  where $z\ge 2$, precluding a conformal field theory (CFT) description of the low-energy physics in the thermodynamic limit, since a CFT requires $z=1$ due to Lorentz invariance~\cite{Bravyi_2012,movassagh2017entanglement}. Recently, the upper-bound on the gap for $s=1$ was refined to $z \geq \frac{5}{2}$, revealing  a sub-diffusive dynamical exponent~\cite{mccarthy2024subdiffusive}.

\begin{table}
\begin{center}
\renewcommand{\arraystretch}{1.2}
\footnotesize
\setlength{\tabcolsep}{5pt}
\begin{adjustbox}{max width=\textwidth}
\begin{tabular}{|c||c|c|c|}
    \hline
    \textbf{Feature} & \textbf{Setting} & \textbf{Result} & \textbf{Reference} \\
    \hline \hline
    Hamiltonian symmetries & Any $s$ & $U(1)^{\times s} \rtimes (\mathbb{Z}_2 \times \mathbb{Z}_2 \times \mathbb{S}_{s})$ & \cref{sec:sym} \\ \hline
    \multirow{6}{*}{\centering \makecell{Entanglement about\\ a cut at $b$}} & \multirow{3}{*}{\centering $s = 1$, $b \gg 1$} &  $S_1 = \frac{1}{2}\log_2(b(1-b/n)) + A_1$ & \multirow{3}{*}{\centering Ref. \cite{movassagh2017entanglement}} \\ 
    & & $S_\kappa = \frac{1}{2} \log_2(b(1-b/n)) + A_\kappa$ & \\
    & & $\chi = b+1$ & \\
    \cline{2-4}
    & \multirow{3}{*}{\centering $s \geq 2$, $b \gg 1$} & $S_1 = 4 \log_2(s) \sqrt{\sigma b(1-b/n)/\pi} + B_1$ & \cref{eq:bipartite-entanglement} \\ 
    & & $S_\kappa = \frac{3}{2 (1-\kappa^{-1})} \log_2(b(1-b/n)) + B_\kappa$ & \cref{eq:renyi-cut} \\
    & & $\chi = \frac{s^{b+1}-1}{s-1}$ & \cref{eq:chi-cut} \\
    \hline
    \multirow{6}{*}{\centering Block entanglement} & \multirow{3}{*}{\centering $s = 1$, $b \gg L^2 \gg 1$} & $S_1 = \frac{1}{2} \log_2(L) + C_1$ & \multirow{3}{*}{\centering Ref. \cite{movassagh2017entanglement}} \\ 
    & & $S_\kappa = \frac{1}{2} \log_2(L) + C_\kappa$ & \\
    & & $\chi = 2L+1$ & \\
    \cline{2-4}
    & \multirow{3}{*}{\centering $s \geq 2$, $b \gg L^2 \gg 1$} & $S_1 = 4\log_2 (s)\; \sqrt{\frac{\sigma L}{\pi}} + D_1$ & \cref{eq:block-entanglement} \\ 
    & & $S_\kappa = \frac{3}{2(1-\kappa^{-1})} \log_2(L) + D_\kappa$ & \cref{eq:renyi-block} \\
    & & $\chi = \frac{(L+1)\, s^{L+2} - (L+2)\, s^{L+1}+1}{(s-1)^2}$ & \cref{eq:chi-block} \\ \hline
    \multirow{2}{*}{\centering $\expval*{S^z_b}$} & $s=1$, $b \gg 1$ & $\frac{2}{\sqrt{3\pi}} \frac{1-2b/n}{\sqrt{b(1-b/n)}} + O(n^{-1})$ & Ref. \cite{movassagh2017entanglement} \\ \cline{2-4}
    & $s \geq 2$, $b \gg 1$ & $(s+1) \sqrt{\sigma/\pi} \frac{1-2b/n}{\sqrt{b(1-b/n)}} + O(n^{-1})$ & \cref{eq:z-s2} \\
    \hline
    $\expval*{S^x_b}$ and $\expval*{S^y_b}$ & Any $s$ and any $b$ & Identically $0$ & \makecell{\cref{sec:sym},\\ Ref.~\cite{movassagh2017entanglement}} \\ \hline
    \multirow{2}{*}{\centering $\expval*{S^z_b S^z_{b+L}}$} & $s = 1$, $b \gg L^2 \gg 1$ & $O(n^{-1})$ & Ref.~\cite{movassagh2017entanglement} \\ \cline{2-4}
    & $s\geq 2$, $b \gg L^2 \gg 1$ & $-\frac{1}{24}\sqrt{\sigma/\pi} (s^2-1)\, L^{-3/2} + O(n^{-1})$ & \cref{eq:zz} \\ \hline
    \multirow{2}{*}{\centering \makecell{$\expval*{S^x_b S^x_{b+L}}$ and\\ $\expval*{S^y_b S^y_{b+L}}$}} & $s=1$, $b \gg L^2 \gg 1$ & $\frac{4}{9} + \frac{7}{18n} + O(n^{-3/2})$ & \cref{eq:xx-s1} \\ \cline{2-4}
    & $s \geq 2$, $b \gg L^2 \gg 1$ & $C_x(s)\, L^{-3/2} + O(n^{-1})$ & \cref{eq:xx-s2} \\
    \hline
    \makecell{$\expval*{S^\alpha_b S^\beta_{b+L}}$;\\ $\alpha,\beta \in \qty{x,y,z}$, $\alpha \neq \beta$} & Any $s$ and any $b,L$ & Identically $0$ & \cref{sec:sym} \\ \hline
\end{tabular}
\end{adjustbox}
\end{center}
\caption{Summary of main results. Throughout, $n$ is the length of the chain; $b$ is the distance (in sites) of the entanglement cut, or of the left edge of a block or correlated pair, from the nearest boundary; $L$ is the number of consecutive spins in a bulk block, or the separation between the two spins of a correlation function; and $\sigma \coloneqq \frac{\sqrt{s}}{1+2\sqrt{s}}$. The Schmidt rank $\chi$ counts the non-zero Schmidt coefficients, the 0-R\'enyi entropy $S_0$ is $\log_2 \chi$, $S_1$ is the von Neumann entanglement entropy, and $S_\kappa$ is a general R\'enyi entropy (we assume $\kappa \geq 2$). Each result is valid in the asymptotic regime indicated in the `Setting' column: `$b \gg 1$' means asymptotics in the cut position (at fixed ratio $b/n$), and `$b \gg L^2 \gg 1$' means the block or pair is deep in the bulk with the thermodynamic limit $n \to \infty$ taken \emph{first} (at fixed $L$), followed by the large-$L$ limit; see the discussion of the order of limits at the end of the Introduction. The block-entanglement entropies $S_1, S_\kappa$ for $s \geq 2$ refer to the exact reduced state of the block, whose spectrum we compute in closed form; see \cref{subsec:block}. For the entanglement entropies, the constants $A_\kappa,\ldots,D_\kappa$ are all calculated explicitly, but omitted in the table for brevity. See the corresponding references for the explicit constants.}
    \label{tab:results}
\end{table}

Many questions have so far remained open regarding the entanglement properties for arbitrary bipartitions and block tripartitions of the chain, higher R\'enyi entropies, and the behavior of all two-point spin correlation functions in both the colorless and colorful models. In this work, we answer these questions, deriving asymptotic analytical expressions --- including explicit additive constants, verified against exact numerics --- for entanglement measures and correlation functions of both the colorless and colorful Motzkin spin chains. Our key contributions are summarized in \cref{tab:results} and include:
\begin{enumerate}
    \item A characterization of the symmetry group of the colorless and colorful models. See \cref{sec:sym}. 
    \item Asymptotic expressions for the entanglement entropy, Schmidt rank, and higher Rényi entropies about an arbitrary cut of the chain for the colorful model. See \cref{subsec:bipartite}.
    \item Block entanglement entropy for subsystems of length $L$ in the bulk for the colorful model. In the thermodynamic limit, the von Neumann entanglement entropy is $4\log s \sqrt{\sigma L/\pi} + O(\log L)$, where $\sigma \coloneqq \frac{\sqrt{s}}{2\sqrt{s}+1}$. See \cref{subsec:block}.
    \item Bulk on-site spin operator expectation values and two-point spin correlation functions in the $X$, $Y$, and $Z$ bases between two spins separated by distance $L$ for both colorless and colorful models. Notably, the two-point correlation functions exhibit the following behavior, described in \cref{sec:sym,sec:spin-operator}:
    \begin{itemize}
        \item From symmetry considerations, all cross correlations vanish: for any $\alpha,\beta \in \qty{x,y,z}$, and $\alpha \neq \beta$ we have $\expval*{S^{\alpha}_b S^{\beta}_{b+L}}=0$.
        \item $\expval*{S^z_b S^z_{b+L}} \sim (1-s^2)L^{-3/2}$ in the thermodynamic limit, in stark contrast to the correlations vanishing asymptotically in $n$ in the bulk for the $s=1$ model.
        \item $\expval*{S^x_b S^x_{b+L}} = \expval*{S^y_b S^y_{b+L}} \sim L^{-3/2}$ for $s>1$, while it approaches a constant value $\frac{4}{9}$ for $s=1$.
    \end{itemize}
    As we discuss  in \cref{sec-discussion}, these unusual power-law decays of the correlation functions with separation are not consistent with any known universality class in one dimension, and point to the exotic nature of the Motzkin spin-chain critical point.
\end{enumerate}
We derive these results for both finite chains and in the thermodynamic limit. For the bulk (block and correlation) results, the order of limits matters and is always the following: the thermodynamic limit $n \to \infty$ is taken \emph{first}, at fixed block size or separation $L$, and the large-$L$ asymptotics are extracted afterwards. Concretely, at finite $n$ our expressions hold up to relative corrections in powers of $L^2/n$ (reflecting the diffusive width $\sim \sqrt{n}$ of the walk-height fluctuations), so the bulk universal behavior in $L$ emerges once $n \gtrsim L^2$; taking the limits in the opposite order (e.g., $L$ growing as a fixed fraction of $n$) instead probes the boundary-dominated regime of the bipartite results. 
A significant aspect of our work is the precision achieved in the resulting analytical expressions, including explicit non-universal prefactors and constants. Furthermore, we verify these expressions against exact and approximate numerical calculations. 

Studying highly entangled ground states of spin models not only provides insights into quantum many-body physics but also has potential implications for quantum computing and simulation. As the complexity of quantum systems grows, classical computers face increasing challenges in simulating their behavior efficiently. This limitation raises important questions about the boundary between classical and quantum computational capabilities. The Motzkin spin-chain models, particularly the colorful variants, present an intriguing case study in this context. Their entanglement properties render the standard 1D tensor-network toolkit inefficient: the $\sqrt{n}$ scaling of the half-chain entanglement entropy implies that any matrix-product-state representation of the ground state requires bond dimension $e^{\Omega(\sqrt{n})}$, a superpolynomial (though subexponential) cost. We emphasize that this is a statement about a specific class of methods rather than a proof of hardness of classical simulation in general. Nevertheless, although we are able to analytically derive their ground-state two-point correlation functions in this work, more complicated quantities, such as general $k$-point correlation functions, dynamical correlation functions, or excited state properties, are at present outside the reach of the available analytical and standard numerical tools. 

Quantum simulation offers one way to overcome these limitations: preparing the colorful Motzkin ground state on quantum hardware would enable us to fully characterize, for example, its higher $k$-point correlation functions. At the same time, our analytical results offer a way to \emph{verify} that we have successfully prepared the desired state on a quantum device by comparing the experimentally measured two-point correlations with the scalings we derive in this work. This mix of resistance to standard classical methods and sharp analytical benchmarks positions these models as natural testbeds for quantum simulation; a concrete proposal for realizing Motzkin spin chains with Rydberg-atom arrays has recently appeared~\cite{mukherjee2026quantum}. Such endeavors could reveal physics that is at present inaccessible to classical computation and analytical methods, highlighting the capabilities of quantum simulators in probing complex quantum systems.

The remainder of this paper is structured as follows. In \cref{sec:review}, we review the construction of the ground state and parent Hamiltonian for the Motzkin spin-chains. In \cref{sec:entanglement}, we present asymptotic expressions for entanglement entropies in the colorful chains for bipartitions about an arbitrary cut of the chain as well as for a centered subsystem of length $L$. In \cref{sec:spin-operator}, we evaluate all two-point spin correlation functions in the thermodynamic limit for both colorful and colorless models. We conclude with a discussion and interpretation of our results and point out some interesting open questions.

\section{Ground state and Hamiltonian of Motzkin spin-chains}\label{sec:review}
We first define the unique  ground state of the Motzkin spin-chain models and then describe the corresponding parent Hamiltonians. Consider a chain of $n$ quantum spins, each of which can be in one of $2s+1$ spin states labeled by $-s,\ldots,s$. A product state basis spanning the full Hilbert space of the chain can be constructed by taking tensor products of these local basis states. Each such product state can be uniquely mapped to a `walk' in the 2-D Cartesian plane that starts at the origin $(0,0)$, and takes one of $2s+1$ types of steps:
\begin{itemize}
\item A `flat' step $\rightarrow$ that changes the position of the walker by $(1,0)$. This corresponds to the $\ket{0}$ state of the spin.
\item $s$ types of `up' steps $\nearrow$, each associated with one of $s$ possible colors, that change the position of the walker by $(1,1)$. These correspond to the states $\ket{1}, \ldots, \ket{s}$ (which we also denote as $\ket*{u^1},\ldots,\ket*{u^s}$ and think of $\ket*{u^k}$ as a step ``up'' with color $k$) of the spin.
\item $s$ types of `down' steps $\searrow$, each associated with one of $s$ possible colors, that change the position of the walker by $(1,-1)$. These correspond to the states $\ket{-1}, \ldots, \ket{-s}$ (which we also denote as $\ket*{d^1},\ldots,\ket*{d^s}$ and think of $\ket*{d^k}$ as a step ``down'' with color $k$) of the spin.
\end{itemize}
An `$s$-colored Motzkin walk' is a walk on $n$ steps that obeys the following conditions: 
\begin{itemize}
    
    \item The walk starts at $(0,0)$ and takes $n$ steps of the above types, ending at $(n,0)$,
    
    \item The walk never falls below the $x$-axis, i.e., the position of the walk for any $0\leq x\leq n$ must be $(x,y)$ for $y\geq 0$. 
    
    \item For $s>1$, the walk must also be `color-matched': If the walk contains a $k$-colored up step $u^k$ at position $(b,m)$, then the next step at $x > b$ that occurs at a height $m$ must be a $k$-colored down-step $d^k$. In other words, a step up can only be followed by a step down of the \emph{same} color, another step up (of any color), or a flat step. In the latter case, every step up must be matched such that steps up and down are nested together (see \cref{fig:blocks} for an illustration).
\end{itemize}

The spin-$s$ Motzkin ground state $\gs$ is defined as the equal superposition of all $s$-colored Motzkin walks of length $n$:
\begin{equation}
\gs = \frac{1}{\sqrt{M_{n}}} \sum_{w \in \text{$s$-colored Motzkin walks}} \ket{w},
\end{equation}
where $M_n$ is the number of $s$-colored Motzkin walks. As shown in Ref. \cite{movassagh2016supercritical}, the number of $s$-colored Motzkin-like walks on $n$ steps that has $m$ unmatched up steps (i.e., obeys all the Motzkin conditions, except that it ends at a height $(n,m)$ for some $m>0$) is asymptotically
\begin{equation}\label{eq:asymptotic_Mnm}
    M_{n,m} \approx \frac{(m+1)\, s^{-m/2}}{2\sigma\sqrt{\sigma \pi}\, n^{3/2}} \qty(\frac{\sqrt{s}}{\sigma})^n \exp(-(m+1)^2/(4 \sigma n))\qc \sigma \coloneqq \frac{\sqrt{s}}{2\sqrt{s}+1}.
\end{equation}
Both the ballot factor $(m+1)$ and the centering of the Gaussian at $m+1$ arise together, from the reflection structure of the walk count (\cref{sec-asymptotic_appendix}); the centering affects only additive constants of the entanglement entropies below, but is needed at that level of precision. We leave the $s$ dependence of $M_{n,m}$ implicit as a particular chain will have fixed integer $s$. At $s=1$, $m=0$, this expression reproduces the classical Motzkin-number asymptotic $M_n \approx \frac{3\sqrt{3}}{2\sqrt{\pi}}\, 3^n n^{-3/2}$. We caution that \cref{eq:asymptotic_Mnm} is a Gaussian (central-limit) approximation, valid uniformly for $m = O(\sqrt{n})$; for fixed $m=O(1)$ the factor $(m+1)$ is exact while for $m \gg 1$ it may be replaced by $m$. This distinction matters whenever a sum over $m$ is dominated by $m = O(1)$, as happens for the R\'enyi entropies in \cref{sec:entanglement}. In all but \cref{sec:xx}, the multiplicative terms that do not depend on $m$ are irrelevant, as they are eventually cancelled out by an overall normalization factor. Therefore, everywhere except \cref{sec:xx}, we will use the simpler form $M_{n,m} \sim (m+1) s^{-m/2} e^{-(m+1)^2/(4 \sigma n)}$. The total number of $s$-colored Motzkin walks on $n$ steps is
\begin{equation}
    M_{n} = \sum_{m=0}^{n/2} s^m M_{n/2,m}^2 \approx \int_0^\infty s^m M_{n/2,m}^2 = \frac{1}{2\sigma^2}\sqrt{\frac{\sigma}{\pi}}\;\frac{(2\sqrt{s}+1)^n}{n^{3/2}}.
\end{equation}
The reason the limits of the sum can be extended from $n/2$ to $\infty$ is that $s^m M_{n/2,m}^2$ contains a factor $e^{-m^2/(4 \sigma n)}$, which suppresses the summands so that they vanish at $m \gtrsim \sqrt{n}$. We will repeatedly make use of this technique of approximating sums with integrals and extending the integral limits throughout this work --- a formal justification of this technique based on the Euler-Maclaurin formula and the general saddle-point method is presented in Refs. \cite{movassagh2017entanglement,flajolet2009analytic}.

We now define some notation. We associate a state $\ket*{C_{n,m[\uparrow],\vec{x}}}$ to be the normalized superposition over all $M_{n,m}$ Motzkin-like (i.e., never crossing below the $x$ axis) walks on $n$ steps that have $m$ unmatched up steps, colored by $\vec{x} \in \S^m$, where $\S=\{1,2,\dots,s\}$ is the set of colors. For instance, the walk in the $A$ segment of \cref{fig:blocks} is one example of a walk on $b$ steps with $m=3$ unmatched up steps, colored by $\vec{x}=[\text{blue},\text{blue},\text{red}]$. We define the state $\ket*{C_{n,m[\downarrow],\vec{x}}}$ similarly. Finally, we define $\ket*{C_{n,p[\downarrow],q[\uparrow],\vec{x}}}$ to be the superposition over all walks on $n$ steps that have $p$ unmatched down steps and $q$ unmatched up steps, which are colored by $\vec{x} \in \S^{p+q}$. For instance, the walk on segment $B$ in \cref{fig:blocks} is a walk on $L$ steps with $p=1$ unmatched down step, $q=2$ unmatched up steps, and colored by $\vec{x} = [\text{red},\text{red},\text{blue}]$. The number of distinct walks contained in the state $\ket*{C_{n,p[\downarrow],q[\uparrow],\vec{x}}}$ is $M_{n,p+q}$. The reason for this is that for the purposes of counting Motzkin walks, the orientations of the unmatched steps do not matter, only their total number~\cite{Bravyi_2012} --- that is, $M_{n,m}$ is equal to the number of Motzkin walks with $m$ unmatched steps \emph{of any kind}.

\begin{figure}
    \centering
    \def\svgwidth{\textwidth}
    \graphicspath{{figs/}}
    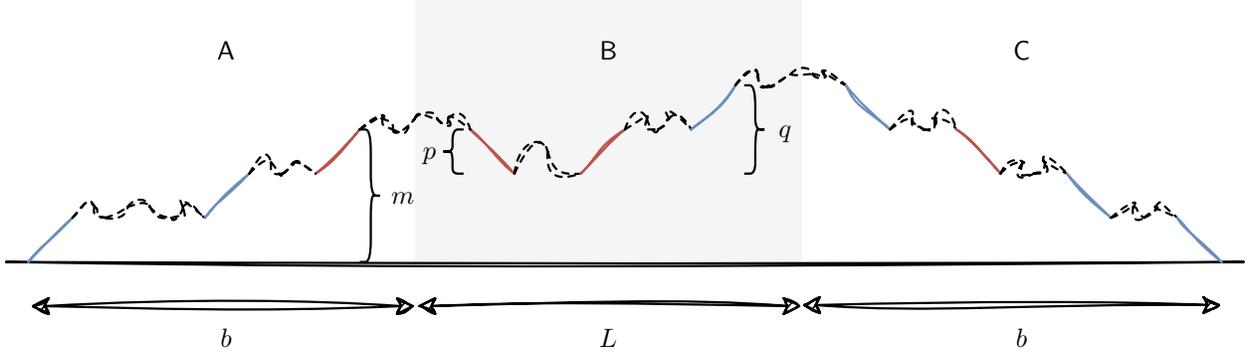
    \caption{Region $A$ shows a Motzkin-like walk on $b$ steps with $m=3$ unmatched up steps. The number of such walks that have the coloring $[\text{blue},\text{blue},\text{red}]$ for the unmatched up steps is $M_{b,3}$. Therefore, the total number of Motzkin-like walks on $b$ steps with $m=3$ unmatched up steps is $s^3 \cdot M_{b,3}$. Region $B$ shows a walk with $p=1$ unmatched down steps and $q=2$ unmatched up steps, and region $C$ shows a Motzkin-like walk with $4$ unmatched down steps. The dashed curves indicate subwalks that are themselves colorful Motzkin walks (e.g., if the sequence $u^1 u^2 0 d^2 d^1$ appeared in the middle of the chain, it would be indicated by a dashed curve).}
    \label{fig:blocks}
\end{figure}

The frustration-free parent Hamiltonian for which the Motzkin state is a ground state can be constructed as a sum of projectors that enforce the constraints defining a valid Motzkin walk. For $s>1$, the Hamiltonian is given by~\cite{movassagh2016supercritical}
\begin{equation} \label{hamiltonian_main}
H = \Pi^{\text{boundary}} + \sum_{j=1}^{n-1} \Pi_{j,j+1} + \sum_{j=1}^{n-1} \Pi^{\text{cross}}_{j,j+1}.
\end{equation}
The projectors above are defined by
\begin{equation}\label{eq:projectors}
    \begin{gathered}
        \Pi_{j,j+1} = \sum_{k=1}^{s} \ketbra*{D^k}_{j,j+1} + \ketbra*{U^k}_{j,j+1} + \ketbra*{\phi^k}_{j,j+1}, \\
        \Pi^{\text{cross}}_{j,j+1} = \sum_{i \neq k} \ketbra*{u^k d^i}_{j,j+1}\qc \Pi^{\text{boundary}} = \sum_{k=1}^s (\ketbra*{d^k}_1 + \ketbra*{u^k}_{n}),
    \end{gathered}
\end{equation}
where $\ket*{D^k} \coloneqq \frac{1}{\sqrt{2}} (\ket*{0d^k}-\ket*{d^k0})$, $\ket*{U^k} \coloneqq \frac{1}{\sqrt{2}} (\ket*{0u^k}-\ket*{u^k0})$, $\ket*{\phi^k} \coloneqq \frac{1}{\sqrt{2}} (\ket*{00}-\ket*{u^kd^k})$, and $u^k,d^k$ represent a $k$-colored up and down step, respectively. The form of $\Pi_{j,j+1}$ encodes the fact that any two valid Motzkin walks are connected by a sequence of three possible local moves: swapping the positions of an adjacent $0$ and $d^k$, swapping an adjacent $0$ and $u^k$, and flipping between $00$ and $u^k d^k$. The cross-term $\Pi^{\text{cross}}_{j,j+1}$ enforces the `color-matching' constraint: if an up-step of color $k$ is followed by a down-step, the down step must also be color $k$. Finally, $\Pi^{\text{boundary}}$ enforces the constraint that the walk must start and end at the origin. The Motzkin state $\gs$ is the unique state that is annihilated by all these projectors, and is therefore the (frustration-free) zero-energy ground state of $H$~\cite{movassagh2016supercritical}. 

A class of continuous deformations of the Motzkin Hamiltonians indexed by a parameter $t\in [0, \infty)$, named `$t$-deformed' Motzkin spin chains, was introduced in Ref.~\cite{zhang2017novel}. The case $t=1$ corresponds to the original Motzkin spin-chains studied in this work. These deformations are of particular interest as the Hamiltonians remain frustration-free for any $t$, while exhibiting different ground state entanglement scaling for various ranges of $t$. For $s=1$, it has been shown that the ground state half-chain entanglement entropy is bounded (i.e., obeys an area-law) for any $t\neq 1$, contrasting with logarithmic scaling at $t=1$. For $s\geq 2$, Ref.~\cite{zhang2017novel} demonstrated that the half-chain entanglement entropy is bounded for $t<1$ and exhibits volume law scaling $\sim n$ for $t>1$, with $\sqrt{n}$ scaling at the $t=1$ point. Related work has also studied the scaling of the gap of the deformed Hamiltonians~\cite{andrei2022spin}. The point $t=1$ we consider in this work can thus be viewed as the critical point of a quantum phase transition in $t$, witnessed by the scaling of entanglement entropy at and across the transition.

\subsection{Hamiltonian Symmetries}\label{sec:sym}
We now characterize the symmetry group of the Motzkin Hamiltonian. For clarity, we first briefly recall the notion of the \emph{semi-direct product} of groups as we will use it repeatedly in the following section. For a finite group $G$ with a subgroup $H$ and a normal subgroup $N$, $G$ is said to be the semi-direct product of $H$ acting on $N$, if $G$ can be expressed as the group product $G=HN \coloneqq \{h\cdot n \mid h\in H, n\in N \}$, and $H$ and $N$ have trivial intersection (i.e.,$H\cap N = \{e\}$ where $e$ is the identity group element). It can be shown that this condition implies that the subgroup $H$ has a group action on $N$ \cite{dummit2003abstract}, $\phi: H \rightarrow \mathrm{Aut}(N)$, which for $h\in H$ and $n\in N$ is given by $\phi(h)(n) = h^{-1}\cdot n \cdot h \in N$ since $N$ is normal. Equivalently, we can think of the group $G$ as a group extension of $H$ by $N$, that is, there exists a short exact sequence 
\begin{equation*}
    1 \rightarrow N \rightarrow G \rightarrow H \rightarrow 1
\end{equation*}
We denote $G = N \rtimes H$ to denote the semi-direct product of $H$ acting on $N$. In the context of symmetry groups of physical systems, this structure commonly appears in describing non-commuting symmetries such that a particular subgroup of symmetries transforms non-trivially under the action of other symmetries. We now state a theorem pertaining to the symmetry group of the Motzkin spin-chain Hamiltonians.
\begin{theorem}[Hamiltonian symmetries]\label{thm::sym}
For any $s$, the Motzkin Hamiltonian has
\begin{itemize}[noitemsep]
    \item continuous $U(1)^{\times s} = U(1)\times \stackrel{s}{\cdots} \times U(1)$ symmetry represented by $e^{-\iu \sum_k \theta_k Q^k}$ for $\theta_k \in [-\pi, \pi]$, generated by the conserved charges $\{Q^k \coloneqq \sum_j (\ketbra*{u^k}_j-\ketbra*{d^k}_j) \mid k=1,\ldots,s\}$, 
    \item $\mathbb{Z}_2$ symmetry represented by the anti-unitary complex conjugation operator $K$,
    \item $\mathbb{Z}_2$ symmetry represented by the space reflection (parity) operator $R \ket{x_1 \ldots x_n}= \ket{x_n \ldots x_1}$ followed by a spin inversion $F$, defined as $F \ket{m} = \ket{-m}$, and
    \item $\mathbb{S}_s$ symmetry represented by global permutations of the color labels, where $\mathbb{S}_{n}$ denotes the symmetric group of $n$ elements, and $\mathbb{S}_{1}$ is the identity.
\end{itemize}
The symmetry group of the Hamiltonian has the structure 
\begin{equation}
    U(1)^{\times s} \rtimes  (\mathbb{Z}_2 \times \mathbb{Z}_2 \times \mathbb{S}_s).
\end{equation}
\end{theorem} 

The unitary $e^{-\iu \sum_k \theta_k Q^k}$ can be interpreted as `selective' rotations of $k$-colored spins about the $z$-axis by angle $\theta_k$. By setting $\theta_k=k \theta$, we see that this symmetry is a generalization of a more conventional $U(1)$ symmetry associated with global rotations about the $z$-axis by angle $\theta$, generated by the conserved charge $S^z_{\text{tot}} = \sum_j S^z_j$. The model also obeys a $\mathbb{Z}_2$ \textit{parity-time} (PT) symmetry. Recall that in spin-systems, the action of time-reversal symmetry is complex conjugation of the wave-function together with inversion of the magnetic quantum number $\ket{m}\rightarrow \ket{-m}$. Conjugation acts on spins through an anti-unitary representation $K$ such that $K$ leaves $S^x$ and $S^z$ invariant, and $K S^y K^{-1} = -S^y$. Therefore, $RFK$ corresponds to the action of the parity-time symmetry,  since $FK$ is a time-reversal operator and $R$ corresponds to a spatial parity transformation. Recalling that the Motzkin Hamiltonian has a \emph{unique} ground state, this still does not contradict Kramer's theorem (which states that eigenstates of half-integer spin Hamiltonians with an anti-unitary symmetry are twofold degenerate), since the Motzkin spin-chains are integer spin models. We now prove \cref{thm::sym}.

\begin{proof}[Proof of \cref{thm::sym}]
We first prove the $U(1)^{\times s}$ symmetry. For this, it suffices to show that the charges $Q^k \coloneqq \sum_j \qty(\ketbra{u^k}_j-\ketbra{d^k}_j)$ for $k=1,\ldots,s$ are conserved. In \cref{eq:projectors}, each Hamiltonian term is diagonal (and hence commutes with each of the charges $Q^k$) except for the terms in $\Pi_{j,j+1}$. However, we immediately see that $Q^\ell \ket{D^k}=Q^\ell \ket{U^k} = Q^\ell \ket{\phi^k}=0$ for any $k,\ell$, which demonstrates $\comm{H}{Q^\ell}=0$ for any $\ell$.

Invariance of the Hamiltonian under complex conjugation follows since the Hamiltonian is manifestly real. The second $\mathbb{Z}_2$ symmetry factor generated by the $RF$ operator (simultaneous reflection and spin inversion) follows from noticing that all projectors in the Hamiltonian in \cref{eq:projectors} are invariant under the replacement $\ket*{a^k b^k} \rightarrow \ket*{\bar{b}^k\bar{a}^k}$, where $a,b\in \{u,0, d\}$ and $\bar{u} = d$, $\bar{d}=u$, and $\bar{0}=0$ --- this is precisely the action of $RF$. Intuitively, this symmetry is reflected in the ground-state by the fact that reading the sequence of steps in a Motzkin walk from right to left results in an upside-down Motzkin walk. 

Next, it is clear from the definition of the Hamiltonian that global permutations of the color labels leave the Hamiltonian invariant. Specifically, for any permutation $\pi \in \mathbb{S}_{s}$ of the $s$ colors, one can define an associated unitary representation $W_\pi$ by its action on the basis states: $W_\pi \ket*{u^k} = \ket*{u^{\pi(k)}}$, $W_\pi \ket{0}= \ket{0}$, and $W_\pi \ket*{d^k} = \ket*{d^{\pi(k)}}$, and the Hamiltonian projectors in \cref{eq:projectors} are invariant under conjugation by this unitary.

It remains to justify the semi-direct product structure of the symmetry group. Since the action of the $U(1)^{\times s}$ generalized rotations has no non-trivial intersection with the action of the other symmetry factors, it suffices to show that the $U(1)^{\times s}$ subgroup is normal~\cite{dummit2003abstract}. We first show that the $U(1)^{\times s}$ subgroup is mapped to itself by action of the $\mathbb{Z}_2 \times \mathbb{Z}_2$ subgroup, and then similarly for the $\mathbb{S}_s$ action. This would show that the $U(1)^{\times s}$ subgroup is normal, since the $\mathbb{Z}_2 \times \mathbb{Z}_2$ and $\mathbb{S}_s$ actions commute. 

First note that the $U(1)^{\times s}$ symmetry action commutes with the action of the $\mathbb{Z}_2$ subgroup corresponding to the $PT$ symmetry since $RFK \cdot e^{-\iu \sum_k \theta_k Q^k} \cdot (RFK)^{-1} = e^{-\iu \sum_k \theta_k Q^k}$. This is because $K$ commutes with the $Q^k$ operators since they are real, while it flips the sign of $\iu$. $F$ negates this since $F Q^k F^\dagger = -Q^k$. The space parity operator $R$ has no effect on $ Q^k$ since all sites are summed over. However, the $U(1)^{\times s}$ symmetry action does not commute with either $K$ or $RF$ independently. Instead, $R F  \cdot e^{-\iu \sum_k \theta_k Q^k} \cdot (RF)^\dagger = e^{\iu \sum_k \theta_k Q^k}$ and $K \cdot e^{-\iu \sum_k \theta_k Q^k} \cdot K^{-1} = e^{\iu \sum_k \theta_k Q^k}$. This shows that the $U(1)^{\times s}$ subgroup is mapped to itself under conjugation by the $\mathbb{Z}_2 \times \mathbb{Z}_2$ action since $e^{\iu \sum_k \theta_k Q^k}\in U(1)^{\times s}$. Finally, we consider the action of $\mathbb{S}_s$ color-permutations on a $U(1)^{\times s}$ rotation. This action simply permutes the $s$ different $U(1)$ factors. Concretely, for $\pi \in \mathbb{S}_s$ with representation $W_\pi$, we have $W_\pi e^{-\iu \sum_k \theta_k Q^k} W_\pi^\dagger = e^{-\iu \sum_k \theta_k Q^{\pi(k)}}$ which is also in $U(1)^{\times s}$. This concludes the argument, showing that the symmetry group of the Hamiltonian has the form  
\begin{equation*}
    U(1)^{\times s} \rtimes  (\mathbb{Z}_2 \times \mathbb{Z}_2 \times \mathbb{S}_s)
\end{equation*}
\end{proof}

\begin{remark*}
    Since the symmetry group of the Hamiltonian is non-Abelian resulting from the non-commuting semi-direct product structure, we expect the excited states to exhibit degeneracies. For $s=1$ and $s=2$, Mackey's theorem from representation theory \cite{etingof2009introduction} can be used to find the non-trivial irreducible representations (irreps) of the semidirect product of the two Abelian groups $U(1)^{\times s}$ and $\mathbb{Z}_2 \times \mathbb{Z}_2 \times \mathbb{S}_s$, allowing for the evaluation of all degeneracies of the model in principle. However, this does not tell us which irrep a particular excited state resides in. For the colorful models with $s\geq 3$, since $\mathbb{S}_{s}$ is itself non-Abelian, one expects the excited states to exhibit additional degeneracy due to the effect of the non-trivial irreps of $\mathbb{S}_{s}$~\cite{robinson1938representations} on the irreps of the semi-direct product. Mackey's theorem does not apply to semi-direct products by a non-Abelian group, and a more specific study of representations of this group is required. In either case, we leave the question of precise evaluation of the models' spectral degeneracies open, as it is beyond the scope of this work.
\end{remark*}
Since $\gs$ is the \emph{unique} ground state of the Hamiltonian, it must also be invariant under the action of the symmetries described in \cref{thm::sym}. These symmetries immediately lead to constraints on certain correlation functions of the ground state, as the following corollary describes.

\begin{corollary*}\label{prop:corr}
Let $i\in[n]$ denote the location of a spin. Then, $\expval*{S^y_i}=\expval*{S^x_i}=0$, $\expval*{S^x_i S^x_{i+L}} = \expval*{S^y_i S^y_{i+L}}$, and $\expval*{S^{\alpha}_i S^{\beta}_{i+L}}=0$  for $\alpha,\beta \in \qty{x,y,z}$ and $\alpha \neq \beta$, in the ground state of the colorless and colorful Motzkin spin-chains.
\end{corollary*}
\begin{proof}
Recall that the $U(1)^{\times s}$ symmetry implies a conventional $U(1)$ spin-rotation symmetry about the $z$-axis. This $U(1)$ symmetry implies that $\expval*{S^x_i} = \expval*{S^y_i}$, and $\expval*{S^x_i S^x_{i+L}} = \expval*{S^y_i S^y_{i+L}}$ for the ground state. Furthermore, since $K$ maps $S^y_i \to -S^y_i$, we can conclude $\expval*{S^y_i} = -\expval*{S^y_i} \implies \expval*{S^y_i}=\expval*{S^x_i}=0$ on any spin for all $i$. Similarly, can infer that all cross-correlations of the form $\expval*{S^{\alpha}_i S^{\beta}_{i+L}}$ vanish identically for $\alpha,\beta \in \qty{x,y,z}$ and $\alpha \neq \beta$. For instance, $\expval*{S^x_{i} S^y_{i+L}}=-\expval*{S^x_{i} S^y_{i+L}}$, allowing us to conclude $\expval*{S^x_i S^y_{i+L}}=0$. An identical argument establishes that $\expval*{S^z_i S^y_{i+L}}=0$. We can then apply the $U(1)$ symmetry to conclude $\expval*{S^z_i S^x_{i+L}} = 0$. 
\end{proof}
Thus, by symmetry considerations alone, we see that all cross-correlations of the spin operators vanish in the ground state of the Motzkin chains. 

\subsection{Spontaneous symmetry breaking}\label{sec::ssb}

The rich symmetry structure of Motzkin spin-chains naturally leads to questions about spontaneous symmetry breaking (SSB) in the thermodynamic limit. Since the spectrum of the Motzkin spin-chain Hamiltonian is gapless, it is particularly intriguing to investigate whether any of its symmetries such as those described in \cref{sec:sym} are spontaneously broken as the spectrum collapses onto the ground-state. A standard diagnostic for SSB is the presence of long-range order (LRO) in the ground-state correlation functions of local operators charged under the symmetry of interest \cite{beekman2019introduction}. Specifically, for a charged local operator $\mathcal{O}$, we expect symmetry breaking if:
\begin{equation}
\lim_{\abs{i-j}\to \infty}\expval{\mathcal{O}_i \mathcal{O}_j} - \expval{\mathcal{O}_i} \expval{\mathcal{O}_j} \neq 0,
\end{equation}
where the expectation is taken with respect to the unique ground-state. The limit $\abs{i-j}\to \infty$ in this case is to be understood as extending the size of the unique ground-state while neglecting the degeneracies that occur from the gap closing. For brevity, we denote the connected correlation function as $\expval*{\mathcal{O}_i \mathcal{O}_j}_c \coloneqq \expval*{\mathcal{O}_i \mathcal{O}_j} - \expval*{\mathcal{O}_i} \expval*{\mathcal{O}_j}$.

For the colorless model, the spin-operator in the $x$-direction $S^x$ is charged under the $U(1)$ spin-rotation symmetry. Thus, we focus on $\expval*{S^x_i S^x_j}_c$ to assess potential $U(1)$ SSB. Importantly, $U(1)$ SSB in this context would not contradict the Mermin-Wagner-Hohenberg-Coleman (MWHC) theorem, which typically precludes continuous SSB in one dimension, even at zero temperature \cite{coleman1988aspects}. The MWHC theorem assumes that continuous symmetry breaking is accompanied by linearly dispersive gapless Nambu-Goldstone modes, which is indeed the case for Lorentz-invariant one-dimensional spin-chains. However, the theorem can be circumvented in one-dimensional, gapless, frustration-free systems, which generally have excitations with quadratic or softer dispersion \cite{watanabe2023spontaneous, gosset2016local, anshu2020improved}, as is known for the Motzkin spin chain \cite{movassagh2016supercritical}.

In \cref{sec:xx}, we demonstrate that for the $s=1$ chain, the ground-state $\expval*{S^x_i S^x_j}_c$ correlation function asymptotes to a constant value 4/9, indicating $U(1)$ SSB in the thermodynamic limit. Notably, this symmetry breaking only occurs as $n\rightarrow \infty$, as we have explicitly shown that $\expval*{S_x} = 0$ for all finite $n$, suggesting an Anderson tower-of-states structure \cite{anderson2018basic, beekman2019introduction} in the low-lying spectrum leading to $U(1)$ SSB. We further test $U(1)$ SSB numerically through exact diagonalization calculations of the ground-state response of the average expectation value $\frac{1}{n}\sum_i \expval*{S^x_i}$ to a perturbing field $h \sum_i S^x_i$, where $h$ is the field strength. SSB is characterized by non-commuting limits of the field-strength and system size:
\begin{equation}\label{eq::limits_ssb}
    \lim_{h\rightarrow 0}\lim_{n\rightarrow \infty} \frac{1}{n}\sum_i \expval*{S^x_i}_{h} \neq \lim_{n\rightarrow \infty}\lim_{h\rightarrow 0} \frac{1}{n}\sum_i \expval*{S^x_i}_{h},
\end{equation}
where $\expval*{\cdot}_h$ denotes the expectation value with respect to the ground state of the perturbed Hamiltonian. \cref{fig:ssb} illustrates the sharpening discontinuity in the response at $h=0$ as $n$ increases, indicating the failure of the limits in \cref{eq::limits_ssb} to commute.

\begin{figure}
    \centering
    \includegraphics[width=0.5\textwidth]{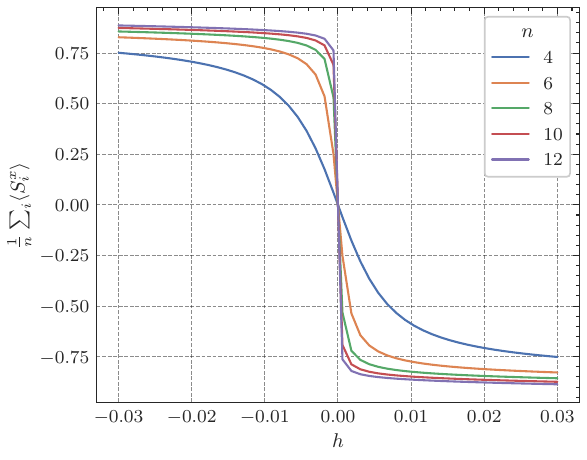}
    \caption{Numerical exact diagonalization calculation of the response of $\frac{1}{n}\sum_i \expval*{S_i^x}$ to a weak perturbing field $h\cdot \sum_iS^x_i$ in the ground-state of the (perturbed) $s=1$ Motzkin Hamiltonian in \cref{hamiltonian_main}. The sharpening of the response curves with increasing system size indicates spontaneous breaking of the $U(1)$ spin-rotation symmetry in the colorless ($s=1$) model.}  
    \label{fig:ssb}
\end{figure}

The colorful chains $(s\geq 2)$ exhibit markedly different behavior. As shown in \cref{sec:xx}, their $XX$ correlation functions decay algebraically as $\sim |i-j|^{-3/2}$, suggesting that the $U(1)$ symmetry corresponding to the diagonal sub-group of the $U(1)^{\times s}$ spin-rotation symmetries is not spontaneously broken. We hypothesize that the constraints induced by the `color-matching' conditions in the colorful model protect the system from SSB, instead leading to `quasi-long-range' order typically associated with a topological Berezinskii–Kosterlitz–Thouless (BKT) type quantum phase transition \cite{kosterlitz6ordering}. This is further discussed in \cref{sec-discussion}.

Finally, we do not expect the discrete $\mathbb{Z}_2$ spin-inversion + parity symmetry or the $\mathbb{S}_s$ color permutation symmetry to be spontaneously broken in either the colorless or colorful models. This is supported by our findings that the $\expval*{S^z_i S^z_j}$ correlation function vanishes with system size in the colorless model and decays algebraically as $\sim \abs{i-j}^{-3/2}$ in the colorful model. Since $S^z$ is charged under both the $\mathbb{Z}_2$ and $\mathbb{S}_s$ symmetries, this behavior suggests that neither is spontaneously broken. Moreover, the preservation of the $\mathbb{S}_s$ symmetry implies that if the diagonal $U(1)$ subgroup remains unbroken, none of the individual $U(1)$ symmetry factors in $U(1)^{\times s}$ can exhibit SSB independently. The $\mathbb{Z}_2$ complex conjugation symmetry is never broken as the Hamiltonian is real.

\section{Entanglement measures}\label{sec:entanglement}
The entanglement entropy of a subsystem is a key quantity that characterizes the amount of quantum correlations between the subsystem and its complement. Previous results on the colorless ($s=1$) Motzkin model revealed that the entanglement entropy scales as $S = \frac{1}{2}\log L + O(1)$ for $L$ consecutive spins in the bulk, and the half-chain entanglement entropy scales as $S = \frac{1}{2}\log n + O(1)$~\cite{Bravyi_2012,movassagh2017entanglement}. It was also shown that, remarkably, the colorful ($s \geq 2$) Motzkin ground state has a half-chain entanglement entropy that to the leading order scales as $\sqrt{n}$~\cite{movassagh2016supercritical}. This exponential violation of the area law is surprising for `natural' spin systems, where `natural' could mean physically realizable or physical models whose Hamiltonians are local, not fine-tuned, and have unique ground states. In this section, we generalize the aforementioned results and derive the scaling of the entanglement entropy for arbitrary bipartitions of the chain by partitioning Motzkin walks by the height they reach in the middle of the chain, and evaluating resulting combinatorial expressions for the Schmidt coefficients. To derive the block entanglement entropy, we express the ground state as a tripartite decomposition of the chain into three segments of sizes $b$, $L$, and $b$, with $\sqrt{b} \gg L$. The primary insight leading to an orthogonal tripartite decomposition of the ground-state is to partition the Motzkin walks according to both their height at the endpoints of the block, as well as the number of unmatched (including color matching) steps up and down in the middle segment of length L. By analyzing the asymptotics of the resulting combinatorial expressions, we are able to extract the universal scaling forms of the entanglement entropy, and precisely characterize the non-universal prefactors and subleading corrections. 

\subsection{Bipartite entanglement about an arbitrary cut}\label{subsec:bipartite}
Ref. \cite{movassagh2016supercritical} showed the asymptotic bipartite entanglement with respect to a bipartition that cuts the chain exactly in half to be
\begin{equation}
    S = 2\log_2(s) \sqrt{\frac{\sigma n}{\pi}} + \frac{1}{2} \log_2 n + \qty(\gamma-\frac{1}{2}) \log_2 e + \log_2(2 \sqrt{\sigma \pi}),
\end{equation}
where $\sigma \coloneqq \sqrt{s}/(2\sqrt{s}+1)$ and $\gamma \approx 0.577$ is the Euler-Mascheroni constant. We first generalize these results by calculating the bipartite entanglement with respect to an \emph{arbitrary} cut.

Consider splitting the chain into segments of length $b$ and $n-b$. Let us label the length-$b$ segment $A$, and label the other segment $B$. Without loss of generality, we assume that $b\leq n-b$. We define $\S$ to be the set of all $s$ colors. The Motzkin state then has a Schmidt decomposition 
\begin{equation}
    \gs \propto \sum_{m=0}^b \sum_{\vec{x} \in \S^m} \sqrt{M_{b,m} M_{n-b,m}} \ket*{C_{b,m[\uparrow],\vec{x}}}_A \otimes \ket*{C_{n-b,m[\downarrow],\vec{\Bar{x}}}}_B.\label{eq:bipartition}
\end{equation}
This is a consequence of the fact that every Motzkin walk can be written as a walk on $A$ with $m$ unmatched up steps, which have colors $\vec{x} = (x_1,x_2,\cdots,x_m)$, followed by a walk on $B$ with $m$ unmatched down steps with colors that are uniquely determined by the colors on the steps in $A$ --- to satisfy the color matching condition, the coloring on $B$ must be the reversed coloring on $A$: $\vec{\Bar{x}} = (x_m, x_{m-1}, \cdots, x_1)$. To see that $\{\ket*{C_{b,m[\uparrow],\vec{x}}} \mid \vec{x} \in \S^m, m=0,\ldots,b\}$ and $\{\ket*{C_{n-b,m[\downarrow],\vec{\Bar{x}}}} \mid \vec{x} \in \S^m, m=0,\ldots,b\}$ are each orthonormal sets, we note that we can associate an equivalence class of walks (each of which is in one-to-one correspondence with a computational basis state) to each of the $\ket*{C_{b,m[\uparrow],\vec{x}}}$, defined as all walks on $b$ steps which have $m$ unmatched up steps colored by $\vec{x}$. These sets are all disjoint, which implies that $\braket*{C_{b,m'[\uparrow],\vec{x}'}}{C_{b,m[\uparrow],\vec{x}}}=\delta_{m,m'} \delta_{\vec{x},\vec{x}'}$. Identical reasoning applies for $\ket*{C_{n-b,m[\downarrow], \vec{x}}}$. Therefore, \cref{eq:bipartition} is a valid Schmidt decomposition from which the Schmidt coefficients can be read off. The following theorem derives the various entanglement entropies associated with this bipartite state.

\begin{theorem}[Bipartite entanglement]
The von Neumann entanglement entropy of the $s\geq 2$ ground state in \cref{eq:bipartition} is asymptotically (in $b$):
\begin{equation}\label{eq:bipartite-entanglement}
    S_1 = 4 \log_2(s) \sqrt{\frac{\sigma \beta}{\pi}} + \frac{1}{2}\log_2 \beta + \qty(\gamma-\frac{1}{2}) \log_2 e + \log_2(2 \sqrt{\sigma \pi}) - \log_2 s + o(1),
\end{equation}
where $\beta \coloneqq b(1-b/n)$, $\sigma \coloneqq \sqrt{s}/(2\sqrt{s}+1)$ and $\gamma \approx 0.5772$ is the Euler-Mascheroni constant. The R\'enyi entanglement entropies $S_\kappa$ for $\kappa \geq 2$ are asymptotically
\begin{equation}\label{eq:renyi-cut}
    S_\kappa = \frac{3}{2(1-\kappa^{-1})} \log_2(\beta) + \frac{\kappa \log_2(2 \sqrt{\pi} \sigma^{3/2}) - \log_2 T_\kappa(s)}{\kappa-1} + O(\beta^{-1})\qc T_\kappa(s) \coloneqq \sum_{m=0}^\infty \frac{(m+1)^{2\kappa}}{s^{(\kappa-1)m}},
\end{equation}
where the constant $T_\kappa(s) = s^{\kappa-1}\, \mathrm{Li}_{-2\kappa}\qty(s^{-(\kappa-1)})$ is expressible via the polylogarithm $\mathrm{Li}$ of negative integer order (a rational function of $s^{\kappa-1}$). Finally, the Schmidt rank is
\begin{equation}\label{eq:chi-cut}
    \chi = \frac{s^{b+1}-1}{s-1}.
\end{equation}
\end{theorem}

\begin{proof}
We will describe the set of Schmidt coefficients by a probability distribution
\begin{equation}\label{eq:cut-dist}
    \Pr(m,\vec{x}) = \frac{1}{Z} M_{b,m}M_{n-b,m}.
\end{equation}
We use the asymptotic form of the Motzkin numbers \eqref{eq:asymptotic_Mnm} to evaluate the marginal distribution
\begin{equation}\label{eq:prob-m}
    \Pr(m) = Z^{-1} \sum_{\vec{x} \in \S^m} M_{b,m} M_{n-b,m} \approx \frac{(m+1)^2 e^{-(m+1)^2/(4 \sigma \beta)}}{2\sqrt{\pi} (\sigma \beta)^{3/2}}\qc \beta \coloneqq b(1-b/n)
\end{equation} 
As shown in \cref{app:bipartite}, the normalization constant of this distribution is $Z = 2\sqrt{\pi}(\sigma \beta)^{3/2}$ to leading order in $\beta$. Since the von Neumann entanglement entropy $S_1$ is the Shannon entropy $\H[m,\vec{x}]$ of the eigenvalue distribution \eqref{eq:cut-dist}, we proceed by evaluating $\H[m,\vec{x}]$. Using the chain rule for conditional entropies, we have $\H[m,\vec{x}] = \H[m] + \H[\vec{x} \mid m]$, and since $\vec{x}$ is uniformly distributed over $s^m$ possibilities given $m$, $\H[\vec{x} \mid m] = \mathbb{E}[m \log_2 s]$. Therefore $S_1 = \E_m[-\log_2 \Pr(m) + m\log_2 s]$. In the variable $u \coloneqq m+1$ the distribution \eqref{eq:prob-m} is a Maxwell distribution, $\Pr(u) = \frac{1}{Z} u^2 e^{-u^2/(4\sigma\beta)}$, and the Maxwell moments $\E[u] = 4\sqrt{\sigma\beta/\pi}$, $\E[u^2] = 6\sigma\beta$, and $\E[\log_2 u] = \frac{1}{2}\log_2(\sigma\beta) + (1-\frac{\gamma}{2})\log_2 e$ assemble $S_1 = \E[u]\log_2 s - \log_2 s + \frac{\E[u^2]}{4\sigma\beta}\log_2 e - 2\,\E[\log_2 u] + \log_2 Z$ into \eqref{eq:bipartite-entanglement}; the additive constant $-\log_2 s$ arises from $m = u - 1$ in the color-entropy term, and the factor $2$ from the two Motzkin factors in \eqref{eq:cut-dist}. We have verified \eqref{eq:bipartite-entanglement} against exact evaluations of the full Schmidt spectrum at $n = 10^4$: the residual error is below $2 \times 10^{-2}$ bits for $s \in \{2,3\}$ at $b = n/2$. The calculation for the R\'enyi entropy is different in an important way: there the sum over $m$ is dominated by $m = O(1)$, where the Gaussian (integral) approximation is not valid and the discrete sum must be kept, giving the constant $T_\kappa(s)$; the details can be found in \cref{app:bipartite}. To conclude, the Schmidt rank is the size of the support of $\Pr(m,\vec{x})$: $\chi = \sum_{m=0}^b s^m = \frac{s^{b+1}-1}{s-1}$, hence the 0-Renyi entropy (i.e., $\log_2 \chi$) is $S_0 \approx (b+1) \cdot \log_2(s) - \log_2(s-1)$ for $s>1$.
\end{proof}

To illustrate the accuracy of the above asymptotics, we present numerical evidence in \cref{fig:renyi}, which demonstrates  agreement of the analytical expression \eqref{eq:renyi-cut} with the exact Renyi entropy computed from the distribution in \cref{eq:cut-dist}.

\begin{figure}
    \centering
    \includegraphics[width=0.8\textwidth]{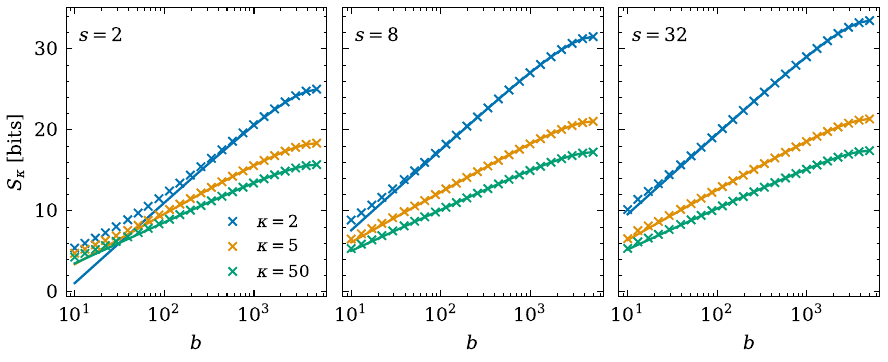}
    \caption{We calculate the exact R\'enyi entropy $S_\kappa(b)$ for a Motzkin spin chain with $n=10^4$ for $s=2,8,32$ and $\kappa=2,5,50$ from the distribution in \cref{eq:cut-dist}, evaluated with exact transfer-matrix values of the Motzkin numbers. The exact results are marked with crosses, while the asymptotic predictions from \cref{eq:renyi-cut}, including the additive constant $-\log_2 T_\kappa(s)/(\kappa-1)$, are shown as solid lines; the residual at $b = n/2$ is below $0.05$ bits in all cases shown.}
    \label{fig:renyi}
\end{figure}

\subsection{Block entanglement}\label{subsec:block}

We turn to the problem of calculating the entanglement entropy of a block of length $L$ in the middle of the colorful ($s\geq2$) Motzkin chain. Assuming that the length of the full chain is $n$, we partition the chain into segments $A$, $B$, $C$, of lengths $b$, $L$, $b$ respectively, so that $b=\lfloor (n-L)/2 \rfloor$.

The key insight that allows us to calculate the block-entanglement of the colorful Motzkin chain is to sum over the number of unmatched down and up steps in $B$, denoted by $p$ and $q$ respectively. This way, the $p$ unmatched steps down in $B$ must be matched by (and uniquely determine the color of) the right-most $p$ unmatched up steps up in $A$, and similarly the $q$ unmatched steps up in $B$ must be matched by (and uniquely determine the color of) the left-most $q$ unmatched down steps down in $C$. There are $s^{p+q}$ ways to color these steps. The coloring of these steps will be encoded by a vector $\vec{x} \in \S^{p+q}$. The remaining number of unmatched steps are those which remain to be matched between $A$ and $C$. If the walk reaches height $m$ at the end of $A$, then there are $m-p$ pairs of steps between $A$ and $C$ that are unmatched within $A$ and $C$ independently. We will encode the colorings of these steps with $\vec{y} \in \S^{m-p}$. 

The ground state $\gs$ can be written in the following tripartite form:
\begin{equation}\label{eq:tripartite}
    \sum_{p+q \leq L}\sum_{m=p}^b  \sum_{\substack{\vec{x} \in \S^{p+q} \\ \vec{y} \in \S^{m-p}}} \sqrt{M_{b,m} M_{L,p+q} M_{b,m+q-p}}  \ket*{C_{b,m[\uparrow],\vec{x},\vec{y}}}_A \otimes \ket*{C_{L,p[\downarrow],q[\uparrow],\vec{x}}}_B \otimes \ket*{C_{b,(m+q-p)[\downarrow],\vec{x},\vec{y}}}_C
\end{equation}
where here the second subscript on $M_{j,k}$ denotes the total of $k$ unmatched steps on a Motzkin walk segment of length $l$. 

An important subtlety is that \cref{eq:tripartite} is \emph{not} a proper Schmidt-type decomposition: while the $B$-side states $\{\ket*{C_{L,p[\downarrow],q[\uparrow],\vec{x}}}\}$ are mutually orthogonal, the $AC$-side states appearing in different $(p,q)$ sectors are not. For example, the pair of segments $A = 0\,0\,u$ and $C = d\,0\,0$ (for $L=2$) completes both $B = 0\,0$ (sector $p=q=0$) and $B = d\,u$ (sector $p=q=1$) into valid Motzkin walks, so the corresponding $AC$ superpositions have non-zero overlap. Consequently, the reduced density matrix $\rho_B = \Tr_{AC} \ketbra{\gs}$ carries coherences between different $(p,q)$ sectors, and any sector-diagonal treatment of \cref{eq:tripartite} misses them. The coherences can nonetheless be resolved exactly. First, $\rho_B$ is \emph{exactly} block-diagonal in $\Delta \coloneqq q - p$: a matrix element between $B$-segments requires a common completing $(A,C)$ pair, which fixes the net height change across the block. Second, within a fixed-$\Delta$ block the completions of a $B$-segment depend only on its \emph{interface class} --- the split $(p,q)$ and the colors of its unmatched steps --- and two classes have common completions only when one is \emph{nested} in the other (the deeper class's unmatched pattern extends the shallower one's, with matching colors). Under this nesting order the classes organize into a forest of complete $s$-ary trees, and in the bulk limit the class weights decay geometrically with nesting depth at rate exactly $1/s$, which collapses the spectrum of each tree onto a single mode per node (\cref{app:block}). The outcome is an explicit spectrum for the \emph{exact} reduced state: the eigenvalues of $\rho_B$ are labeled by an integer $t \geq 0$ with
\begin{equation}\label{eq:rho-spectrum}
    \lambda(t) \propto (t + A)\, s^{-t}\, e^{-(t+A)^2/(4\sigma L)} \qc \operatorname{mult}(t) = \qty(1 - \frac{1}{s^2})\, s^{t}\,(t + A) \quad (t \geq 1) \qc A \coloneqq \frac{s^2+1}{s^2-1},
\end{equation}
and $\operatorname{mult}(0) = 1$. The Gaussian window is centered at $t + A$ rather than $t$: the ballot factor of $M_{L,\eta}$ centers each sector's window at $\eta + 1$ (\cref{sec-asymptotic_appendix}), and averaging over the depth of each tree shifts the center by the further $A - 1$ (\cref{app:block}). The centering affects only the additive constants of the entropies, but with it, direct evaluation from \cref{eq:rho-spectrum} reproduces the closed forms below with no further corrections. The following theorem characterizes the block entanglement of the exact reduced state $\rho_B$.
\begin{theorem}[Block entanglement]\label{thm:block}
The von Neumann entropy of the block reduced state $\rho_B$ is
\begin{equation}\label{eq:block-entanglement}
    S_1 = 4 \log_2 s \sqrt{\frac{\sigma L}{\pi}} + \log_2 L + \qty(\frac{\gamma+1}{2}) \log_2 e + \log_2(2\sigma\sqrt{\pi}) - A \log_2 s - \log_2\qty(\frac{s^2}{s^2-1}) + o(1).
\end{equation}
The $\kappa$-R\'enyi entropies of $\rho_B$ for $\kappa \geq 2$ are
\begin{subequations} \label{eq:renyi-block}
\begin{align}
    S_\kappa &= \frac{3}{2(1-\kappa^{-1})} \log_2(L) + \frac{\kappa \log_2(2\sqrt{\pi} \sigma^{3/2}) - \log_2 V_\kappa(s)}{\kappa-1} + O(L^{-1/2}), \\
    V_\kappa(s) &\coloneqq \qty(\frac{s^2}{s^2-1})^{\kappa} \qty[A^\kappa + \qty(1 - \frac{1}{s^2})\sum_{t=1}^\infty \frac{(t+A)^{\kappa+1}}{s^{(\kappa-1)t}}].
\end{align}
\end{subequations}
Finally, the Schmidt rank is
\begin{equation}\label{eq:chi-block}
    \chi = \sum_{\eta=0}^{L} (\eta+1)\, s^\eta = \frac{(L+1)\,s^{L+2} - (L+2)\,s^{L+1} + 1}{(s-1)^2}.
\end{equation}
\end{theorem}
\begin{proof}
Substituting $u \coloneqq t + A$ in \cref{eq:rho-spectrum}, the factor $s^{-t}$ cancels against the $s^t$ in the multiplicity, and the total probability carried by the modes at level $t$ becomes
\begin{equation}
    \operatorname{mult}(t)\,\lambda(t) = \frac{1}{Z}\, u^2\, e^{-u^2/(4\sigma L)} \qc Z \approx \int_0^\infty u^2 e^{-u^2/(4\sigma L)} \dd{u} = 2\sqrt{\pi}\,(\sigma L)^{3/2},
\end{equation}
a Maxwell distribution in $u$ (replacing the sum by the integral and extending its lower limit from $A$ to $0$ incurs only $o(1)$ errors, the window being $u \sim \sqrt{\sigma L}$). The normalized eigenvalues are $\lambda = \frac{s^2}{s^2-1}\, \frac{1}{Z}\, u\, s^{-(u-A)}\, e^{-u^2/(4\sigma L)}$, so
\begin{equation}
    S_1 = \E[-\log_2 \lambda] = \E[u] \log_2 s + \frac{\E[u^2]}{4\sigma L} \log_2 e - \E[\log_2 u] + \log_2 Z - A \log_2 s - \log_2\qty(\frac{s^2}{s^2-1}),
\end{equation}
and the Maxwell moments $\E[u] = 4\sqrt{\sigma L/\pi}$, $\E[u^2] = 6\sigma L$, and $\E[\log_2 u] = \frac{1}{2}\log_2(\sigma L) + (1 - \frac{\gamma}{2})\log_2 e$ assemble to \cref{eq:block-entanglement}. For the R\'enyi entropies ($\kappa \geq 2$) the power sum $\sum_t \operatorname{mult}(t)\,\lambda(t)^\kappa$ instead localizes at $t = O(1)$, where $e^{-\kappa (t+A)^2/(4\sigma L)} = 1 - O(L^{-1})$; the sum must be kept discrete, and evaluates to $V_\kappa(s)$ times the normalization factors, giving \cref{eq:renyi-block} (see \cref{app:block}). We have verified \cref{eq:block-entanglement,eq:renyi-block} against exact numerical spectra (radial construction with exact transfer-matrix inputs; see \cref{app:block}): the residuals decay as $O(L^{-1/2})$ and at $L \sim 10^5$ are ${\approx}\,5\times10^{-3}$ bits for $S_1$ and below $10^{-4}$ bits for $S_2, S_5$, for $s \in \{2,3\}$. Finally, $\rho_B$ acts identically on the $B$-segments within each interface class (\cref{app:block}), so its support is spanned by the $(\eta+1)\, s^\eta$ class-uniform vectors at each level $\eta \in \{0,\ldots,L\}$, all of which carry nonzero weight at finite $b$ (verified by exact diagonalization at small sizes); hence
\begin{equation}
    \chi = \sum_{\eta=0}^L (\eta+1) \cdot s^\eta = \frac{(L+1)\,s^{L+2}-(L+2)\,s^{L+1}+1}{(s-1)^2},
\end{equation}
yielding a 0-R\'enyi entropy $S_0 \approx L \cdot \log_2(s) + \log_2(L) + \log_2\qty(\frac{s}{s-1})$ for $s>1$.
\end{proof}
\begin{remark*}
Treating \cref{eq:tripartite} as a Schmidt decomposition --- i.e., dephasing $\rho_B$ over the $(p,q)$ sectors --- replaces the spectrum \eqref{eq:rho-spectrum} by $\lambda(\eta) \propto (\eta+1)\, s^{-\eta}\, e^{-(\eta+1)^2/(4\sigma L)}$ with multiplicity $(\eta+1)\,s^\eta$. In the variable $u = \eta+1$ this has the identical Maxwell structure, so dephasing leaves every $L$-dependent term of \cref{eq:block-entanglement,eq:renyi-block} unchanged, but overestimates the additive constants: the von Neumann entropy by exactly
\begin{equation}\label{eq:coherence-cost}
    \frac{2\log_2 s}{s^2 - 1} + \log_2\qty(\frac{s^2}{s^2-1}) \approx \begin{cases} 1.08 \text{ bits}, & s = 2,\\ 0.57 \text{ bits}, & s = 3,\end{cases}
\end{equation}
vanishing as $s \to \infty$, and the $\kappa$-R\'enyi entropies by $\frac{1}{\kappa-1}\log_2(V_\kappa(s)/U_\kappa(s))$ with $U_\kappa(s) \coloneqq \sum_{\eta=0}^\infty (\eta+1)^{\kappa+1} s^{-(\kappa-1)\eta}$. The overestimate is guaranteed in sign a priori, since the dephased state is a pinching of $\rho_B$ and is therefore majorized by it~\cite{nielsen2001majorization}; the Schmidt rank is unaffected by dephasing.
\end{remark*}

\section{Spin operator expectation values and correlation functions}\label{sec:spin-operator}

We turn our attention to spin operator expectation values $\expval*{S^\alpha_i}$ and two-point spin correlation functions $\expval*{S^\alpha_i S^\beta_j}$ ($\alpha,\beta \in \qty{x,y,z}$). For the colorless Motzkin chain ($s=1$), it was previously shown that the $\expval*{S^z_i S^z_j}$ correlations vanish with system size as $O(n^{-1})$ in the bulk of the chain, and it was conjectured that similar scaling would hold for the colorful case~\cite{movassagh2017entanglement}. The results of this work surprisingly contradict this conjecture, showing a nontrivial power law decay of $\expval*{S^z_i S^z_j} \sim \abs{i-j}^{-3/2}$ in the bulk of the colorful chain. Furthermore, in previous work, the correlations in the $X$ and $Y$ bases were not calculated for either the $s=1$ or $s \geq 2$ models. Here we complete this picture by deriving analytical expressions for all two-point correlation functions $\expval*{S^\alpha_i S^\beta_j}$ for arbitrary $s$.

The key insight is to relate the spin-spin correlation functions to certain probability distributions on the space of colorful Motzkin walks. For the $S^z S^z$ correlations, we define a `height' operator $H_b=\sum_{i=1}^b S^z_i$ whose expectation value within a given Motzkin walk measures the total spin in the $Z$ direction up to the point $b$. We first derive analytical expressions for $\expval*{H_b}$ and $\expval*{H_b H_{b+L}}$. Then, $S^z_b$ can be expressed as a derivative of $H_b$, and likewise, $S^z_{b} S^z_{b+L}$ can be expressed as $\expval*{S^z_b S^z_{b+L}} = \dv{b}\dv{(b+L)} \expval*{H_b H_{b+L}}$.
For the calculation of $\expval*{H_b H_{b+L}}$, we analyze the joint probability distribution of the number of unmatched up and down steps in each segment. By making use of the asymptotic formulas for the number of Motzkin walks with a fixed unmatched height, we derive an integral formula for $\expval*{H_b H_{b+L}}$ which can be evaluated exactly in the limit $1\ll L \ll \sqrt{n}$.

For the in-plane correlations (e.g., $S^xS^x$ or $S^yS^y$), the situation is more subtle. Since the model and ground-state are $U(1)$ symmetric and \cref{prop:corr} showed that all cross-correlations vanish, it suffices to calculate $\expval*{S^x_i S^x_j}$. The $S^x$ operator is not directly related to the height of the Motzkin walks. Rather, it induces transitions between different walks. By analyzing the action of the operator $S^x_b S^x_{b+L}$  on an arbitrary Motzkin walk and identifying all conditions under which it maps a valid Motzkin walk to another valid Motzkin walk, we derive a closed-form expression for $\expval*{S^x_b S^x_{b+L}}$ as a sum of probabilities of different transition events. 

\subsection{Expectation value of \texorpdfstring{$S^z$}{Sz}}

In \cref{prop:corr}, it was shown that the expectation value of on-site operators $S^x$ and $S^y$ vanishes identically. We now show that the expectation value of $S^z$ vanishes with system size in the thermodynamic limit, as the following theorem elucidates.

\begin{theorem}[Local magnetization]
For a site at position $b$ in the bulk of the chain ($b \gg 1$), the expectation value $\expval*{S^z_b}$ for both the colorless and colorful Motzkin spin-chains asymptotically satisfies
\begin{equation}\label{eq:Sz_exp}
\expval*{S^z_b} = (s+1) \sqrt{\frac{\sigma}{\pi}} \frac{1-\frac{2b}{n}}{\sqrt{b(1-b/n)}}
\end{equation}
\end{theorem}
\begin{proof}
We define a sequence of operators $H_b \coloneqq \sum_{i=1}^b S^z_i$ as the cumulative spin up to a point $b$. Since $S^z_b = S^z_b - S^z_{b-1}$, the strategy will be to evaluate $\expval*{H_b}$ as a function of $b$, and differentiate with respect to $b$ to find $\expval*{S^z_b}$. A simplifying observation is that $H_b$ is diagonal in $S^z$ basis, and the Motzkin state $\gs$ is an equal superposition of $S^z$ basis product states all with the same phase. Therefore, for the purposes of evaluating $\expval*{H_b}$, we can treat $\gs$ as a uniform \emph{classical} probability distribution over all Motzkin walks on $n$ steps:
\begin{equation}
    \expval*{H_b}{\psi_{\text{GS}}} = \E_{w \sim \text{Unif}(\text{Motzkin walks})}[H_b(w)].
\end{equation}
To evaluate the above expectation, we condition on $m$, the number of unmatched up steps at the point $b$, and $\vec{x} \in \S^m$, which are the colors of the $m$ unmatched up steps, and then iterate the expectation with respect to $m$ and $x$. Each of the colors in $\S$ is associated with a spin value $1,\ldots,s$, and we will write $\abs{\vec{x}}$ to denote the sum of the spin values associated with the colors $\vec{x}$. Now, observe that for a walk $w$ associated with colors $\vec{x}$, $H_b(w) = \abs{\vec{x}}$, since all the other spins before the point $b$ are matched (hence their contributions to the sum cancel). Finally, note that $\E_{\vec{x}}[\abs{\vec{x}} \mid m] = \sum_{i=1}^m \E[x_i \mid m] = \frac{m (s+1)}{2}$, since by symmetry (see \cref{sec:sym}), $x_i$ is uniformly distributed over all $\S$ possibilities, and the average spin $x_i$ is $\frac{s+1}{2}$.
\begin{equation}
    \E_w[H_b(w)] = \E_{m,\vec{x}}[\E_w[H_b(w) \mid m, \vec{x}]] = \E_m[\E_{\vec{x}}[\abs{\vec{x}} \mid m]] = \frac{s+1}{2}\E[m].
\end{equation}
Recalling $\Pr(m)$ from \cref{eq:prob-m},
\begin{equation}\label{eq:z-s2}
    \begin{gathered}
        \E[H_b] \approx \frac{s+1}{2} \int_0^\infty m \Pr(m) \dd{m} = 2(s+1) \sqrt{\sigma \beta/\pi}, \qq{and} \\
        \expval*{S^z_b} = \dv{b} \E[H_b] = \frac{2(s+1) \sqrt{\sigma/\pi}}{\sqrt{\beta}} \dv{\beta}{b} = (s+1) \sqrt{\sigma/\pi} \frac{1-\frac{2b}{n}}{\sqrt{b(1-b/n)}}
    \end{gathered}
\end{equation}
\end{proof}
We have verified \cref{eq:Sz_exp} against exact transfer-matrix evaluations for $s \in \{1,2,3\}$: at $n = 16000$ the ratio of exact to asymptotic values is $1.000$ at $b/n \in \{1/8, 1/4, 3/8, 3/4\}$, with the deviation decaying as $O(n^{-1})$.

\subsection{\texorpdfstring{$\expval*{S^z_i S^z_j}$}{ZZ} correlation function}
We take a similar approach for $\expval*{S^z_i S^z_j}$. Assuming $j>i$, observe that $(H_i-H_{i-1}) (H_j -H_{j-1}) = S^z_i S^z_j$, so $\expval*{S^z_i S^z_j} = \dv{i} \dv{j} \expval*{H_i H_j}$. Therefore it suffices to evaluate $\expval*{H_i H_j}$ as a function of $i$ and $j$. We will work in the bulk of the chain, at points $i=b \gg L^2$ and $j=b+L$. We will label the segment from $1,\ldots,b$ as $A$, the segment $b+1,\ldots,b+L$ as $B$, and the remainder of the chain as $C$ (\cref{fig:blocks}). We first evaluate the conditional expectation $\E_w[H_b(w) H_{b+L}(w) \mid m,p,q,\vec{x},\vec{y}]$, where $m$ represents the number of up unmatched steps on $A$, $p$ and $q$ represent the number of unmatched down and up steps in $B$, $\vec{y} \in \S^m$ represents the coloring of the $m$ unmatched steps, and $\vec{x} \in \S^q$ represents the color of the unmatched up steps in $B$. Note that the color of the $q$ unmatched down steps in $B$ is uniquely determined by the last $q$ colors of $\vec{y}$. 

As in the computation of $\expval*{S^z_b}$, $H_b$ = $\abs{\vec{y}}$. However, since the contribution to the total spin from the $p$ unmatched down steps in $B$ cancel the contribution from the last $p$ unmatched up steps in $A$, we have $H_{b+L}=\abs{\vec{y}_{1:m-p}} + \abs{\vec{x}}$, where the notation $\vec{y}_{1:m-p} \in \S^{m-p}$ indicates the vector that contains entries $1$ through $m-p$ of $\vec{y}$. Therefore, by splitting $\abs{\vec{y}} = \abs{\vec{y}_{1:m-p}} + \abs{\vec{y}_{m-p+1:m}}$, we must evaluate
\begin{subequations}
\begin{gather}
    \E[H_b H_{b+L} \mid m, p,q] = \E_{\vec{x},\vec{y}}\qty[\abs{\vec{y}_{1:m-p}}^2 + \abs{\vec{y}_{1:m-p}} \cdot \abs{\vec{y}_{m-p+1:m}} + \abs{\vec{y}} \cdot \abs{\vec{x}} \mid m,p,q]
\end{gather}
Note that $\vec{y}_{1:m-p}$ and $\vec{y}_{m-p+1:m}$ are independent, so that the expectation factorizes. Recall that $\E_{\vec{y}}[\abs{\vec{y}_{1:m-p}} \mid m,p,q] = \frac{s+1}{2}(m-p)$, and $\E_{\vec{y}}[\abs{\vec{y}_{1:m-p}} \mid m,p,q] = \frac{s+1}{2} p$. Similarly, for the third summand, $\vec{y}$ and $\vec{x}$ are independent, so these expectations can be evaluated separately (the expectations are $\frac{s+1}{2} m$ and $\frac{s+1}{2} q$, respectively). To evaluate the expectation of the first summand, observe that $\vec{y}_{1:m-p}$ is uniformly distributed over all possibilities in $\S^{m-p}$. Let $\text{M}(m-p, 1/s)$ be the multinomial distribution over $m-p$ draws where each outcome is chosen with the same probability $1/s$. This models the number of occurrences of each color when $\vec{y}$ is uniformly drawn from $\S^{m-p}$.
\begin{gather}
    \E[H_{b+L} H_b \mid m, p,q] = \E_{\vec{n} \sim \text{M}(m-p; 1/s)}\qty[\qty(\sum_{k=1}^s k \cdot n_k)^2] + \qty(\frac{s+1}{2})^2  \qty((m-p)p + mq) \\
    = (m^2-mp+mq) \qty(\frac{s+1}{2})^2 + \frac{(m-p)(s^2-1)}{12} \label{eq:hh-final}
\end{gather}
\end{subequations}
Note that neither the expectation of $m^2$ nor $m$ depend on the position of the second point $b+L$ (they depend only on $b$). Therefore, upon differentiating with respect to the second position $b+L$, these terms will drop out, and we do not need to evaluate the expectation values of these terms. The following Lemma describes the evaluation of expectation values of the other relevant terms in \cref{eq:hh-final} with respect to the probability distribution $\Pr(m,p,q)$. 

\begin{lemma}\label{lemma:szsz}
Let the segment of $L$ consecutive spins be centered on the chain such that $b+L+b=n$, and $L \ll \sqrt{b}$.The following relations hold for the distribution $\Pr(m,p,q)$ of Motzkin walks with height $m$ and $p$ ($q$) unbalanced down (up) steps.
\begin{subequations}\label{eq:expects}
    \begin{gather}
        \E[m(q-p)] = \frac{L \sigma(3L-n)}{n} \label{eq:mdelta} \\
        \E[-p] = -2\sqrt{\frac{\sigma}{\pi}} \sqrt{L} + O(L^2/n). \label{eq:-p}
    \end{gather}
\end{subequations}
\end{lemma}
\begin{proof}
The distribution over $m$, $p$, and $q$ satisfies
\begin{equation}
    \Pr(m,p,q) \propto s^{m+q} M_{b,m} M_{b,m+q-p} M_{L,p+q}.
\end{equation}
By using the asymptotic forms of the Motzkin numbers, we can manipulate the expectations in \cref{eq:expects} into Gaussian integrals. See \cref{app:zz} for details. 
\end{proof}

From the above lemma, the asymptotic correlation function $\expval*{S^z_i S^z_j}$ can be calculated in the limit $n \rightarrow \infty$ for $0 \ll L^2 \ll n$ as described by the following theorem. 

\begin{theorem}[Colorful $S^z S^z$ correlations]
    The correlation function $\expval*{S^z_{(n-L)/2} S^z_{(n+L)/2}}$ in the thermodynamic limit ($n \rightarrow \infty$), has the following asymptotic expression with respect to $L$
\begin{equation}\label{eq:zz}
   \expval*{S^z_{\frac{n-L}{2}} S^z_{\frac{n+L}{2}}} = -\frac{1}{24}\sqrt{\frac{\sigma}{\pi}}\;(s^2-1)\; L^{-3/2} + O(n^{-1})
\end{equation}
\end{theorem}

Note that this result applies for segments of length $L$ that are not necessarily dead-centered in the middle of the chain, but are sufficiently far away from the boundaries. It can be shown that the off-centering of the segment only gives sub-leading vanishing corrections (in $n$) to the expression of the correlation function in the above theorem.

\begin{proof}
Taking expectations of the terms in \cref{eq:expects} and substituting the result of \cref{lemma:szsz} into \cref{eq:hh-final}, it follows that
\begin{equation}
    f(L) \coloneqq \expval*{H_{i} H_{j}} = -\frac{s^2-1}{6} \sqrt{\sigma/\pi} \sqrt{L} - \qty(\frac{s+1}{2})^2 L \sigma + O(L^2/n); \quad L \coloneqq j-i.
\end{equation}
Then,
\begin{equation}
    \expval*{S_{i}^z S_{j}^z} = \dv{i} \dv{j} \expval*{H_{i} H_{j}} = -\dv[2]{f}{L} = -\frac{\sqrt{\sigma/\pi} (s^2-1)}{24} L^{-3/2} + O(n^{-1}).
\end{equation} 
\end{proof}
We have verified \cref{eq:zz} by exact transfer-matrix evaluation, decomposing the correlator as $\expval*{S^z_b S^z_{b+L}} = \qty(\frac{s+1}{2})^2 \E[\varepsilon_b\, \varepsilon_{b+L}] - \frac{s^2-1}{12} \Pr(\text{partner})$, where $\varepsilon_i \in \{0, \pm 1\}$ is the step direction at site $i$ and $\Pr(\text{partner})$ is the probability that the two sites are a matched pair (an identity following from the independence of colors of distinct pairs, validated against brute-force enumeration). At fixed $L \in \{32, 64\}$, with $n$ up to $2.5 \times 10^7$, the covariance term decays as $n^{-1}$ (successive quadruplings of $n$ reduce it by factors of $3.96$--$3.99$), and the full correlator converges to the $L^{-3/2}$ asymptote to within $2.6\%$ at $L = 32$ and $1.2\%$ at $L = 64$ for $s \in \{2, 3\}$ --- a relative deviation roughly halving as $L$ doubles, consistent with the $O(L^{-5/2})$ subleading correction.

We provide numerical evidence of the validity of this formula in \cref{fig:zz}, where we calculate the exact correlation function based on \cref{eq:hh-final}, showing convergence of the asymptotic prediction in \cref{eq:zz}. Notably, this contradicts a previous conjecture that the scaling $\expval*{S^z_i S^z_j} = O(n^{-1})$ for the colorless Motzkin model would generalize to the colorful case~\cite{movassagh2017entanglement}. However, note that \cref{eq:zz} is still entirely consistent with the known $s=1$ scaling: upon substituting $s=1$, the $L^{-3/2}$ scaling vanishes, and we are left with $\expval*{S^z_i S^z_j} = O(n^{-1})$.

\begin{figure}
    \centering
    \includegraphics[width=0.6\textwidth]{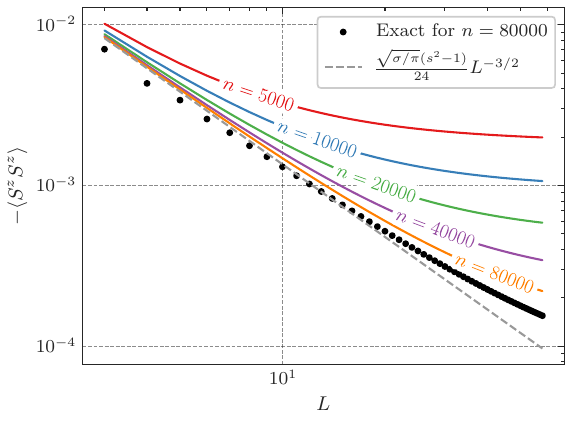}
    \caption{The black dots show $\expval*{S^z_{b} S^z_{b+L}}$ for $n=80000$ based on an exact numerical evaluation of \cref{eq:pmpq}. The colored lines show the approximate theoretical prediction for $\expval*{S^z_{b} S^z_{b+L}}$ based on \cref{eq:pmetadelta}. The dashed grey line shows the asymptotic prediction $\expval*{S^z_b S^z_{b+L}} \propto L^{-3/2}$ \eqref{eq:zz}. All results are shown for $s=2$.}
    \label{fig:zz}
\end{figure}

\subsection{\texorpdfstring{$\expval*{S^x_i S^x_j}$}{XX} correlation function}\label{sec:xx}
The calculation of the $\expval*{S^x_i S^x_j}$ correlation function provides further insight into the quantum correlations present in the Motzkin spin-chain ground state. Unlike the $S^z$ correlations, which can be directly related to the height of the Motzkin walks, the $S^x$ correlations require a more subtle analysis. We begin by examining the action of the $S^x$ operator on the individual spin states and then consider its effect on the entire ground state. The approach will be to express $S^x_i S^x_j \gs$ as a sum of $S^z$ basis product states and identifying which of these correspond to valid Motzkin walks. We consider the colorless and colorful models separately in the following calculations.

\paragraph*{Colorless model ($s = 1$):} 
Note that the $S^x$ operator acts on $s=1$ basis states as
\begin{equation}\label{eq:sx-s1}
    \begin{gathered}
        S^x \ket*{-1} = \ket*{0}/\sqrt{2} \\
        S^x \ket*{0} = (\ket*{-1}+\ket*{1})/\sqrt{2} \\
        S^x \ket*{1} = \ket*{0}/\sqrt{2}
    \end{gathered}
\end{equation}
The state $S^x_i S^x_j \gs$ can be expressed as a sum of walks.
Only some of these walks will be valid Motzkin walks; it suffices to count the number of valid walks, as acting with $\bra{\psi_{\text{GS}}}$ selects only those states which are valid Motzkin walks. 

Let us focus on a particular Motzkin walk $w$, with individual steps labelled $\{w_i\}$ for $i=1\cdots n$. There are two necessary and sufficient conditions for which applying $S^x_i S^x_j$ to this walk will result in an \textit{invalid} Motzkin walk.
\begin{enumerate}
    \item $w_i$ and $w_j$ are both $-1$, or they are both $+1$. In this case, $S^x_i S^x_j$ increases or decreases the net height of the chain by $2$, respectively, rendering the resulting state an invalid Motzkin walk.
    \item $w_i$ is $0$ or $1$, and gets mapped to a lower spin (i.e., $-1$ or $0$), \emph{and} there is some $i \leq k < j$ such that $\sum_{\ell=1}^k w_\ell = 0$. That is, the spin lowering at $i$ caused the walk to become negative in between $i$ and $j$. This will correspond to a {\it rare event} which we denote $E$.
\end{enumerate}

This means that the overall expectation will be
\begin{equation} \label{eq:SxSx_colorless}
2\expval*{S^x_i S^x_j} = \sum_{w_b \neq w_{b+L}} \Pr(w_b,w_{b+L}) + 2\Pr(0,0) - \Pr(w_b\in \qty{0,1} \land E).
\end{equation}
where the notation $\land$ denotes a logical and.
There is a prefactor of $2$ on $\Pr(0,0)$ because $S^x S^x$ maps $(w_i, w_j)=(0,0)$ to two different pairs $(-1,1)$ and $(1,-1)$ that both result in valid Motzkin walks. We proceed by calculating $\Pr(w_b,w_{b+L})$, then showing that the contribution of the last summand $\Pr(w_b\in \qty{0,1} \land E)$ is asymptotically negligible.

\begin{lemma} \label{lemma:SxSxcolorless_1}
    For $w_b, w_{b+L} \in \{-1,0,1\}$ and $w_b \neq w_{b+L}$, to leading order in $n^{-1}$, 
    \begin{equation}\label{eq:p-wb}
        \Pr(w_b,w_{b+L}) \propto 1 - \frac{3(w_b + w_{b+L})^2}{4\sigma n}.
    \end{equation}
\end{lemma}
\begin{proof}[Proof sketch]
We use the exact expression (neglecting normalization of the distribution)
\begin{equation}
    \Pr(w_b, w_{b+L}) \propto \sum_{m=0}^{b} \sum_{\substack{p+q \leq L-2 \\ p \leq m+w_b}} M_{b,m} M_{L-2,p+q} M_{n-b-L,m+w_b+q-p+w_{b+L}} \label{eq:pb-exact},
\end{equation}
which counts the number of Motzkin walks between $b$ and $b+L$, conditioned on the fact that the step at $b$ takes value $w_b$ and the step at $b+L$ is $w_{b+L}$. The sum over $m$ runs over the possible heights the Motzkin walk can reach within $b$ steps, and the sum over $p,q$ runs over the number of unmatched up and down steps between positions $b$ and $b+L$. By simplifying \cref{eq:pb-exact} using the asymptotic form of the Motzkin numbers, we arrive at the desired result \eqref{eq:p-wb}. See \cref{app:s1-corr} for the details of this analysis. In \cref{fig:prwb}, we compare the asymptotic prediction \eqref{eq:p-wb} with the exact form \eqref{eq:pb-exact}. 
\begin{figure}
    \centering
    \includegraphics[width=0.65\textwidth]{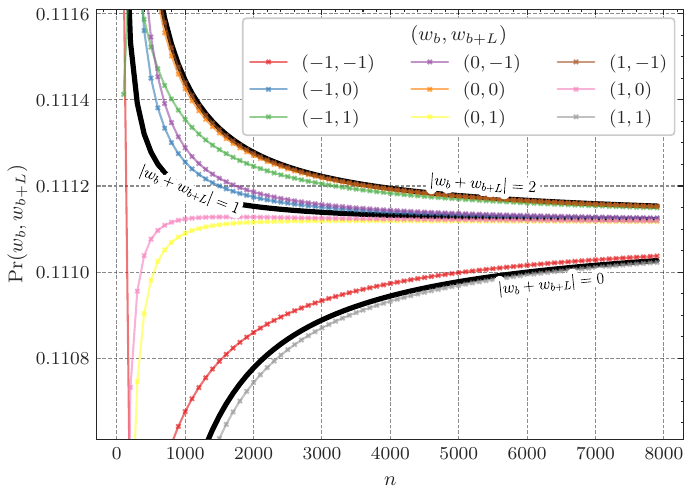}
    \caption{Asymptotic prediction $\Pr(w_b,w_{b+L}) \propto 1-\frac{3(w_b+w_{b+L})^2}{4 \sigma n}$ in solid black lines \eqref{eq:pb-exact} for the three independent cases where $\abs{w_b + w_{b+L}} \in \{0,1,2\}$ respectively (since the asymptotic expression for $\Pr(w_b,w_{b+L})$ only depends on $\abs{w_b+w_{b+L}}$). Exact calculation of all $9$ entries in $\Pr(w_b, w_{b+L})$ are shown in colored dots.}\label{fig:prwb}
\end{figure}
\end{proof}

\begin{lemma} \label{lemma:SxSxcolorless_2}
For any $L \ll \sqrt{n}$,
\begin{equation}
    \Pr(w_b \in \qty{0,1} \land E) = O(n^{-3/2}).
\end{equation}
\end{lemma}
\begin{proof}
We upper bound $\Pr(w_b\in \qty{0,1} \land E)$ with $\Pr(E)$. We will assume $[b,b+L]$ is centered around the middle of the chain. Then,
\begin{equation}
    \begin{aligned}
        \Pr(E) &= \frac{1}{M_n} \sum_{p + q \leq L} M_{b,p} M_{L,p+q} M_{n-b-L,q} \\
        &\lesssim \frac{1}{M_n} \int_{-\infty}^\infty \int_{\abs{\Delta}}^\infty M_{b,(\eta-\Delta)/2} M_{L,\eta} M_{b,(\eta+\Delta)/2} \dd{\eta} \dd{\Delta} \\
        &= \frac{4L\sqrt{\sigma/\pi}}{\sqrt{b}\sqrt{n(b+L)}} = O\qty(\frac{L}{n^{3/2}})
    \end{aligned}
\end{equation}
\end{proof}
From \cref{lemma:SxSxcolorless_1,lemma:SxSxcolorless_2}, the asymptotic $\expval*{S^xS^x}$ correlation function follows immediately.
\begin{theorem}[Colorless $S^x S^x$ correlations]
The correlation function $\expval*{S^x_i S^x_j}$ of the colorless $(s=1)$ Motzkin spin chain in the thermodynamic limit ($n \rightarrow \infty$), for $i$, $j$ such that the $L\coloneqq j-i$ spins are centered around the middle of the chain is asymptotically (in $n$) constant as
\begin{equation}\label{eq:xx-s1}
    \expval*{S^x_i S^x_j} = \frac{4}{9} + \frac{7}{18n} + O(n^{-3/2})
\end{equation}
\end{theorem}
\begin{proof}
From \cref{eq:SxSx_colorless} and the vanishing probability of the event $E$ from \cref{lemma:SxSxcolorless_2}, we conclude that
\begin{equation*}
    \expval*{S^x_i S^x_j} = \frac{1}{2} \sum_{w_b \neq w_{b+L}} \Pr(w_b,w_{b+L}) + \Pr(0,0) + O(n^{-3/2})
\end{equation*}
From \cref{lemma:SxSxcolorless_1} it then follows that 
\begin{equation}
    \expval*{S^x_i S^x_j} = \frac{1}{2}\; \frac{8-\frac{3}{\sigma n}}{9-\frac{6}{\sigma n}} = \frac{4}{9}+\frac{7}{18n} + O(n^{-3/2})
\end{equation}
\end{proof}

\cref{fig:xx} illustrates convergence of the bulk $XX$ correlation function to a constant value as the system size is increased, using numerical Matrix Product State (MPS) simulations of the $s=1$ Motzkin spin-chain \cite{state_prep_to_appear}.

\begin{figure}
    \centering
    \includegraphics[width=0.6\textwidth]{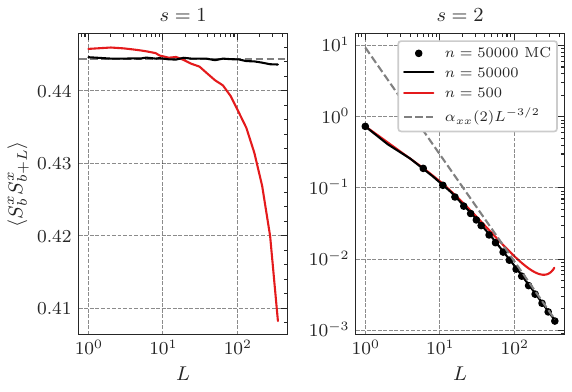}
    \caption{Exact $s=1$ results for $\expval*{S^x_b S^x_{b+L}}$ in the middle of the chain, calculated via Matrix Product State numerics. The dashed black line indicates $\nicefrac{4}{9}$, the asymptotic value of $\expval*{S^x_b S^x_{b+L}}$ predicted by \cref{eq:xx-s1}.} 
    \label{fig:xx}
\end{figure}

\paragraph*{Colorful model ($s \geq 2$):} The  $\expval*{S^x_b S^x_{b+L}}$ correlation function in the colorful model differs significantly from that of the colorless model, as we describe in the following theorem. 
\begin{theorem}[Colorful $S^x S^x$ correlations]\label{thm:xx-s2}
The correlation function $\expval*{S^x_i S^x_j}$ of the colorful $(s\geq 2)$ Motzkin spin chain in the thermodynamic limit ($n \rightarrow \infty$), for $i$, $j$ such that the $L\coloneqq j-i$ spins are centered around the middle of the chain satisfies the asymptotic algebraic decay
\begin{equation}
    \begin{gathered}
    \expval*{S^x_b S^x_{b+L}} = C_x(s)\, L^{-3/2} + O(L^{-5/2}) \qc \text{where} \\
    C_x(s) \coloneqq \frac{1}{2\sqrt{\pi}\,\sigma^{3/2}} \qty[\frac{s^2(s+1)^2}{2\,(2\sqrt{s}+1)^2 (\sqrt{s}-1)^3 (\sqrt{s}+1)} + \frac{(s^2-1)\,\sigma^2}{3}],
    \end{gathered} \label{eq:xx-s2}
\end{equation}
with $C_x(2) \approx 9.17$ and $C_x(3) \approx 4.40$. The coefficient of the subleading correction is numerically large, so the leading term dominates only for $L \gtrsim 10^3$.
\end{theorem}
\begin{proof}[Proof sketch]
Since $S^x$ is tridiagonal in the local basis, $S^x_b S^x_{b+L}$ connects a walk $w$ only to walks $w'$ that differ from $w$ at the two sites, each by one unit of $S^z$. Height conservation and the color-matching (stack) structure leave three families of transitions: exchanging two flat steps for an up--down (or down--up) pair, moving an up or down step from one site to a flat step at the other, and recoloring a matched pair when the two sites are matched partners. In the first two families the intervening segment need \emph{not} stay above the initial height: a dip of depth $m$ below it, ending at height $G$ above the running minimum, displaces the color stack and forces $m+G$ color equalities, each holding with probability $1/s$. The correlator therefore reduces to the walk functional $\E[s^{-(m+G)}]$, which by the local limit law $\Pr(m, G) \simeq (m+G+1)/(2\sqrt{\pi}\,\sigma^{3/2} L^{3/2})$ localizes at $m + G = O(1)$: the color cost reweights the dip depths but preserves the ballot-type $L^{-3/2}$ decay. Summing the geometric series $\sum_{k\ge0} (k+1)^2 s^{-k} = s^2(s+1)/(s-1)^3$ over the channels, weighted by the bulk letter frequencies and the $S^x$ matrix elements, and adding the matched-partner channel, yields \cref{eq:xx-s2}. The bookkeeping is given in \cref{app:s2-corr}, together with numerical verification by exact dynamic programming and by Monte Carlo sampling of colored walks.
\end{proof}

\section{Discussion}\label{sec-discussion}
The results presented in the preceding sections reveal a diverse range of physical phenomena in both the colorless and colorful Motzkin spin chain models. We now turn to interpreting these findings and exploring their broader implications for our understanding of quantum criticality, entanglement, and correlations in these one-dimensional systems.

\paragraph*{\textbf{Hardness of classical simulability:}} The entanglement properties of the Motzkin spin chain models have significant implications for their classical simulability. For the colorless ($s=1$) model, the logarithmic scaling of entanglement entropy implies that it can be efficiently simulated using matrix product states (MPS) with a bond dimension that scales polynomially with system size. In fact, one can find an exact MPS representation analytically for the $s=1$ Motzkin ground state, as we will present in an upcoming work \cite{state_prep_to_appear}. However, the situation is markedly different for the colorful ($s>1$) models. The $\sqrt{n}$ scaling of the half-chain entanglement entropy implies that any MPS representation of the ground state requires a bond dimension $e^{\Omega(\sqrt{n})}$ --- superpolynomial, though subexponential, in system size. We stress that this is a limitation of a specific (albeit standard) class of classical methods, not a proof of general classical hardness: it does not rule out efficient classical algorithms outside the MPS framework. Moreover, it is believed that the model is not exactly solvable in general --- interestingly, a variant of the model dubbed the `free Motzkin spin-chain' was introduced in \cite{hao2023exact} and shown to be integrable with unusual scattering properties. Together, these facts suggest that the standard analytical and 1D tensor-network methods for the colorful Motzkin spin chains become impractical for large systems (especially beyond ground-state properties), making these models interesting testbeds for quantum simulation of many-body quantum systems~\cite{mukherjee2026quantum}. Experimentally corroborating the exact results for the anomalous ground state correlation functions presented in this work can then serve as verification of preparing the correct quantum state, before pursuing simulations beyond the solvable regime that may reveal physics that is inaccessible through classical computational and analytic methods.

\paragraph*{\textbf{Need for large system sizes:}} The calculations in this work suggest that all bulk universal features of the system only appear for system sizes that are larger than the squared size of the region under consideration. For instance, observing the bulk $\sim \sqrt{L}$ block-entanglement scaling for a centered segment of size $L$ requires a system of size $n \geq L^2$, and further assumes that $L \gg 1$. The underlying Brownian motion structure of the walks therefore enforces the need for large system sizes to observe the asymptotic scaling relations derived above. This property poses a challenge for accurate numerical simulations that probe the bulk physics in the thermodynamic limit. This property may also imply the presence of a marginally relevant operator in the renormalization-group sense (or at least an operator that is relevant but has an anomalously small scaling dimension) in a continuum field theory of the system, as significant coarse graining is required to reveal universal scaling (see \cref{fig:zz} for an example of this phenomenon for the $ZZ$ correlation function).

\paragraph*{\textbf{Non-vanishing $\expval*{S^z_i S^z_j}$ correlations in the colorful model:}}
One surprising property of the colorful model revealed by the calculations in this work, that is absent in the colorless model, is the $\sim L^{-3/2}$ power-law decay of the bulk $S^zS^z$ (connected) correlation function in the thermodynamic limit. For the colorless model, it was shown in Ref. \cite{movassagh2017entanglement} that this correlation function vanishes in the bulk as $\sim 1/n$ with system size, independent of $L$. The presence of such correlations in the colorful models indicates that the spins are not in a disordered liquid-like state. This suggests the possibility of an `algebraic' liquid with quasi-long-range order, similar to the phase found in the Gutzwiller-projected Fermi sea state \cite{anderson1987resonating, gutzwiller1965correlation}. However, much more needs to be understood about the physical interpretation of this order.

\paragraph*{\textbf{$U(1)$ spontaneous symmetry breaking in the colorless chain:}} The calculations in this work reveal the structure of the in-plane correlations of both the colorless and colorful chains. The $XX$ and $YY$ correlations are shown to be asymptotically constant in the colorless model, revealing a type of bulk long-range order (LRO), indicating spontaneous symmetry breaking (SSB) of the $U(1)$ spin-rotation symmetry in the thermodynamic limit as noted in \cref{sec::ssb}. Such LRO without non-zero on-site expectation values for any finite system size is related to the violation of the cluster-decomposition property of correlations \cite{hastings2004locality, nachtergaele2006lieb} in the Motzkin chain ground state that was shown in \cite{dell2016violation} --- it is therefore possible to have $\expval*{S^x_i S^x_{i+L}} \neq 0$ in the bulk while $\expval*{S^x} = 0$ everywhere. 

\paragraph*{\textbf{Algebraic (quasi-long-range) order in the colorful chain}:}
In contrast to the colorless model, in the colorful model we find that the in-plane correlations decay algebraically as $\sim L^{-3/2}$, identical to the scaling of $ZZ$ correlations in the ground-state. This feature further supports the interpretation of the colorful chain ground-state as exhibiting an algebraic liquid order. Moreover, the difference in the in-plane correlations between the colorless and colorful chains suggests that the frustration induced by the 'color-matching' conditions cause breaking of long-range order into quasi-long-range order, indicating that the  $U(1)$ spin-rotation symmetry (corresponding to the diagonal subgroup of $U(1)^{\times s}$) is not spontaneously broken in the colorful model. Instead, the algebraic decay of correlations is reminiscent of the scaling of correlation functions at a Berezinskii–Kosterlitz–Thouless (BKT) type topological phase transition, suggesting that the phase transitions associated with $t$-deformations may be BKT-like. This finding is consistent with the BKT-like behaviour of correlations across the $t$-deformation transition found numerically in \cite{barbiero2017haldane}. However, as we discuss in the next sub-section, the $3/2$ exponent in the algebraic decay is not consistent with the critical exponents of the universality class associated with the usual BKT transition in the 2-D XY model, which would instead predict an exponent of $1/4$ \cite{kosterlitz6ordering}.

\paragraph*{\textbf{Anomalous power-law scaling of correlations in the colorful model:}}

For a generic critical or gapless system, the scaling hypothesis \cite{kardar2007statistical} implies that connected correlation functions of operators decay as a power-law
\begin{equation}
    \expval*{O_i O_j} - \expval*{O_i} \expval*{O_j} \propto \abs{i-j}^{-d+2-\eta_O}
\end{equation}
where $d$ is the physical dimension, and $\eta_O$ is the so-called `anomalous dimension' of the operator $O$ at the gapless point. For scaling operators $O$, the exponents $\eta_O$ are part of the data of the theory describing the universality class associated with the low-energy physics of the model. In the colorful Motzkin model, where $d=1$, we have shown that as $n\to\infty$, the bulk $XX$, $YY$, and $ZZ$ correlations scale as $\sim L^{-3/2}$. This implies that $\eta_{S^\alpha} = \frac{5}{2}$ for all $\alpha \in \qty{x,y,z}$. While we do not have a good enough understanding of the operator content of the underlying field theory, we note that this anomalous dimension is not consistent with any known 1D universality class. In particular, these exponents are unusually large, and imply a significant deviation from scale-invariance under renormalization group flow away from the fixed point. A detailed construction of a continuum field theory for the Motzkin chain is required to understand the universality class associated with the gapless point.

\paragraph*{\textbf{Heavy tailed entanglement spectrum:}} A striking feature of the colorful Motzkin models is the heavy-tailed nature of their entanglement spectrum (i.e., the square of the Schmidt coefficients in \cref{eq:bipartition}). This property manifests in a significant discrepancy between the scaling of the von Neumann entropy ($\sim \sqrt{\beta}$ for a cut at $b$, with $\beta = b(1-b/n)$) and the higher R\'enyi entropies ($\sim \log \beta$); the same separation ($\sqrt{L}$ versus $\log L$) holds for a block of $L$ spins, where we compute the spectrum of the exact reduced state in closed form: the inter-sector coherences reduce each entropy only by an $L$-independent constant (\cref{eq:coherence-cost}) that vanishes as $s \to \infty$, so the exponential separation holds for the exact block state with exact additive constants. Such a large difference arises from the presence of many small but non-negligible eigenvalues in the reduced density matrix, which contribute substantially to the von Neumann entropy but are known to be suppressed by the $\kappa > 1$ R\'enyi entropies~\cite{muller2013quantum}. To the best of our knowledge, this represents the first explicit quantum many-body model where such an exponential separation between entropy measures has been analytically proven. Similar behavior was also shown for Motzkin and Fredkin spin chains in Ref.~\cite{Sugino_2018}. The heavy-tailed entanglement spectrum underscores the exotic nature of the quantum correlations in colorful Motzkin chains and suggests that different entanglement measures may capture distinct aspects of the model's physics.

There are many open questions and directions for future research that we believe are worth exploring:
\begin{enumerate}
    \item The construction of a continuum quantum field theory that captures the universal long wavelength properties of the Motzkin chains. While a continuum description of the $s=1$ ground state was introduced in \cite{chen2017gapless}, a systematic construction of a field-theoretic Hamiltonian and its low energy spectrum remains to be understood for both the colorless and colorful models. However, the symmetries of the Hamiltonian and correlation function scaling relations unveiled in this work strongly constrain the underlying field theory --- we intend to study this in future work.
    
    \item Calculating dynamical physical observables such as the spectral function, dynamic structure factors, and out-of-time-ordered correlation functions. These would provide valuable insight into the structure of excitations and the (presumably) chaotic dynamics in these models. For example, numerical work using the Density Matrix Renormalization Group (DMRG) method in \cite{chen2017quantum} has shown the existence of low energy excitations with multiple emergent time-scales and dynamical critical exponents in the $s=1$ spin chain. 
    
    \item Exploring completely translation invariant versions of the Motzkin spin chains. The boundary projectors in the Hamiltonian in \cref{hamiltonian_main} violate true translation invariance, although the bulk of the Motzkin ground state is indeed translation invariant in the thermodynamic limit. A translation invariant version of the Motzkin chain that uses a uniform external magnetic field was proposed in \cite{movassagh2016supercritical}. However, this model requires a magnetic field with magnitude $O(n^{-1})$ making it fine-tuned and unrealistic for experimental implementations. Nevertheless, exploring the phenomenology of this model as well as developing other translation invariant variants of the model that are as highly entangled remains open.
    
    \item The generalization of these results to Motzkin loop models in higher dimensions. An initial exploration of this idea in 2D has been developed in Ref. \cite{balasubramanian20232d}. The combinatorial tools developed in this paper could be adapted to analyze more general lattice walk models in higher dimensions.

    \item An extension of our results on correlation functions and entanglement measures to the half-integer spin variants of the Motzkin spin-chains, named \textit{Fredkin} spin chains \cite{salberger2018fredkin, salberger2017deformed} would follow immediately from modifying the calculations in this work to count 'Dyck walks' instead of Motzkin walks (see \cref{sec-asymptotic_appendix}). It would be interesting to explore whether there are significant differences in the asymptotic scaling relations of correlations, symmetries, and symmetry breaking between Fredkin and Motzkin spin chains.
    \item Exploring experimental realizations and quantum simulations: While Motzkin spin chains are theoretical constructs, recent advances in quantum simulation platforms, such as Rydberg atom arrays, show promise for realizing and controlling highly entangled quantum states \cite{koyluoglu2024floquet}, and a concrete proposal for quantum simulation of Motzkin spin chains with Rydberg atoms has recently been put forward~\cite{mukherjee2026quantum}. These developments underscore the relevance of understanding models like Motzkin spin chains for quantum simulation. An important step towards potential realization is the development of efficient quantum circuits for both preparing Motzkin ground states and simulating time evolution under the Motzkin Hamiltonian. Determining the minimum circuit depth for state preparation and time evolution could enable simulations on programmable quantum processors \cite{ebadi2021quantum, bernien2017probing}. We plan on addressing these questions in our upcoming work \cite{state_prep_to_appear}.
\end{enumerate}

\acknowledgments{We thank Mikhail Lukin, Susanne Yelin, Subir Sachdev, Kun Yang, Nicholas Bonesteel, Nishad Maskara, Nazli Ugur Koyluoglu, Pablo Bonilla, Rohith Sajith, and Nik Gjonbalaj for many valuable discussions.}

\section*{Author contributions}
V.M. and A.G. conceived the project, developed the analytical methods, performed the calculations, and wrote the manuscript. R.M. advised on and supervised the work, and checked the technical results for correctness.

\bibliographystyle{quantum}
\bibliography{refs}

\begin{thebibliography}{10}

\bibitem{movassagh2016supercritical}
Ramis Movassagh and Peter~W Shor.
\newblock ``Supercritical entanglement in local systems: Counterexample to the area law for quantum matter''.
\newblock \href{https://dx.doi.org/10.1073/pnas.1605716113}{Proceedings of the National Academy of Sciences {\bf 113}, 13278--13282}~(2016).

\bibitem{calabrese2004entanglement}
Pasquale Calabrese and John Cardy.
\newblock ``Entanglement entropy and quantum field theory''.
\newblock \href{https://dx.doi.org/10.1088/1742-5468/2004/06/p06002}{Journal of Statistical Mechanics: Theory and Experiment {\bf 2004}, P06002}~(2004).

\bibitem{zhang2017novel}
Zhao Zhang, Amr Ahmadain, and Israel Klich.
\newblock ``Novel quantum phase transition from bounded to extensive entanglement''.
\newblock \href{https://dx.doi.org/10.1073/pnas.1702029114}{Proceedings of the National Academy of Sciences {\bf 114}, 5142--5146}~(2017).

\bibitem{levine2017gap}
Lionel Levine and Ramis Movassagh.
\newblock ``The gap of the area-weighted {Motzkin} spin chain is exponentially small''.
\newblock \href{https://dx.doi.org/10.1088/1751-8121/aa6cc4}{Journal of Physics A: Mathematical and Theoretical {\bf 50}, 255302}~(2017).

\bibitem{movassagh2017entanglement}
Ramis Movassagh.
\newblock ``Entanglement and correlation functions of the quantum {Motzkin} spin-chain''.
\newblock \href{https://dx.doi.org/10.1063/1.4977829}{Journal of Mathematical Physics{\bf 58}}~(2017).

\bibitem{Bravyi_2012}
Sergey Bravyi, Libor Caha, Ramis Movassagh, Daniel Nagaj, and Peter~W. Shor.
\newblock ``Criticality without frustration for quantum spin-1 chains''.
\newblock \href{https://dx.doi.org/10.1103/physrevlett.109.207202}{Physical Review Letters{\bf 109}}~(2012).

\bibitem{mccarthy2024subdiffusive}
Catherine McCarthy, Hansveer Singh, Sarang Gopalakrishnan, and Romain Vasseur.
\newblock ``Subdiffusive transport in the {Fredkin} dynamical universality class''.
\newblock \href{https://dx.doi.org/10.1103/PhysRevB.111.184317}{Physical Review B {\bf 111}, 184317}~(2025).
\newblock  \href{http://arxiv.org/abs/2407.11110}{arXiv:2407.11110}.

\bibitem{mukherjee2026quantum}
Kaustav Mukherjee, Hatem Barghathi, Adrian Del~Maestro, and Rick Mukherjee.
\newblock ``Quantum simulation of {Motzkin} spin chains with {Rydberg} atoms''.
\newblock \href{https://dx.doi.org/10.1103/ph91-6pg2}{Physical Review Letters {\bf 137}, 150601}~(2026).

\bibitem{flajolet2009analytic}
P.~Flajolet and R.~Sedgewick.
\newblock ``Analytic combinatorics''.
\newblock \href{https://dx.doi.org/10.1017/CBO9780511801655}{Cambridge University Press}. ~(2009).

\bibitem{andrei2022spin}
Radu Andrei, Marius Lemm, and Ramis Movassagh.
\newblock ``The spin-one {Motzkin} chain is gapped for any area weight $t<1$''~(2022).
\newblock  \href{http://arxiv.org/abs/2204.04517}{arXiv:2204.04517}.

\bibitem{dummit2003abstract}
David~S. Dummit and Richard~M. Foote.
\newblock ``Abstract algebra''.
\newblock John Wiley \& Sons. Hoboken, NJ~(2004).
\newblock 3rd edition.

\bibitem{etingof2009introduction}
Pavel Etingof, Oleg Golberg, Sebastian Hensel, Tiankai Liu, Alex Schwendner, Dmitry Vaintrob, and Elena Yudovina.
\newblock ``Introduction to representation theory''.
\newblock \href{https://dx.doi.org/10.1090/stml/059}{Volume~59 of Student Mathematical Library}.
\newblock American Mathematical Society. Providence, RI~(2011).
\newblock  \href{http://arxiv.org/abs/0901.0827}{arXiv:0901.0827}.

\bibitem{robinson1938representations}
G.~de~B.~Robinson.
\newblock ``On the representations of the symmetric group''.
\newblock \href{https://dx.doi.org/10.2307/2371609}{American Journal of Mathematics {\bf 60}, 745--760}~(1938).

\bibitem{beekman2019introduction}
Aron Beekman, Louk Rademaker, and Jasper van Wezel.
\newblock ``An introduction to spontaneous symmetry breaking''.
\newblock \href{https://dx.doi.org/10.21468/scipostphyslectnotes.11}{SciPost Physics Lecture Notes}~(2019).

\bibitem{coleman1988aspects}
Sidney Coleman.
\newblock ``{Aspects of Symmetry}: {Selected Erice Lectures}''.
\newblock \href{https://dx.doi.org/10.1017/CBO9780511565045}{Cambridge University Press}. Cambridge, U.K.~(1988).

\bibitem{watanabe2023spontaneous}
Haruki Watanabe, Hosho Katsura, and Jong~Yeon Lee.
\newblock ``Critical spontaneous breaking of {U(1)} symmetry at zero temperature in one dimension''.
\newblock \href{https://dx.doi.org/10.1103/PhysRevLett.133.176001}{Physical Review Letters {\bf 133}, 176001}~(2024).
\newblock  \href{http://arxiv.org/abs/2310.16881}{arXiv:2310.16881}.

\bibitem{gosset2016local}
David Gosset and Evgeny Mozgunov.
\newblock ``Local gap threshold for frustration-free spin systems''.
\newblock \href{https://dx.doi.org/10.1063/1.4962337}{Journal of Mathematical Physics{\bf 57}}~(2016).

\bibitem{anshu2020improved}
Anurag Anshu.
\newblock ``Improved local spectral gap thresholds for lattices of finite size''.
\newblock \href{https://dx.doi.org/10.1103/physrevb.101.165104}{Physical Review B{\bf 101}}~(2020).

\bibitem{anderson2018basic}
Philip~W. Anderson.
\newblock ``Basic notions of condensed matter physics''.
\newblock \href{https://dx.doi.org/10.4324/9780429494116}{CRC Press}. ~(2018).

\bibitem{kosterlitz6ordering}
J~M Kosterlitz and D~J Thouless.
\newblock ``Ordering, metastability and phase transitions in two-dimensional systems''.
\newblock \href{https://dx.doi.org/10.1088/0022-3719/6/7/010}{Journal of Physics C: Solid State Physics {\bf 6}, 1181}~(1973).

\bibitem{nielsen2001majorization}
Michael~A Nielsen and Guifr{\'e} Vidal.
\newblock ``Majorization and the interconversion of bipartite states''.
\newblock \href{https://dx.doi.org/10.26421/QIC1.1-5}{Quantum Inf. Comput. {\bf 1}, 76--93}~(2001).

\bibitem{state_prep_to_appear}
Andi Gu, Varun Menon, and Ramis Movassagh.
\newblock ``State preparation and simulation of quantum {Motzkin} spin chains''.
\newblock to appear.

\bibitem{hao2023exact}
Kun Hao, Olof Salberger, and Vladimir Korepin.
\newblock ``Exact solution of the quantum integrable model associated with the {Motzkin} spin chain''.
\newblock \href{https://dx.doi.org/10.1007/jhep08(2023)009}{Journal of High Energy Physics{\bf 2023}}~(2023).

\bibitem{anderson1987resonating}
P.~W. Anderson, G.~Baskaran, Z.~Zou, and T.~Hsu.
\newblock ``{Resonating--valence-bond theory of phase transitions and superconductivity in ${\mathrm{La}}_{2}$${\mathrm{CuO}}_{4}$-based compounds}''.
\newblock \href{https://dx.doi.org/10.1103/PhysRevLett.58.2790}{Phys. Rev. Lett. {\bf 58}, 2790--2793}~(1987).

\bibitem{gutzwiller1965correlation}
Martin~C. Gutzwiller.
\newblock ``Correlation of electrons in a narrow $s$ band''.
\newblock \href{https://dx.doi.org/10.1103/PhysRev.137.A1726}{Phys. Rev. {\bf 137}, A1726--A1735}~(1965).

\bibitem{hastings2004locality}
M.~B. Hastings.
\newblock ``Locality in quantum and {Markov} dynamics on lattices and networks''.
\newblock \href{https://dx.doi.org/10.1103/physrevlett.93.140402}{Physical Review Letters{\bf 93}}~(2004).

\bibitem{nachtergaele2006lieb}
Bruno Nachtergaele and Robert Sims.
\newblock ``{Lieb-Robinson} bounds and the exponential clustering theorem''.
\newblock \href{https://dx.doi.org/10.1007/s00220-006-1556-1}{Communications in Mathematical Physics {\bf 265}, 119–130}~(2006).

\bibitem{dell2016violation}
L.~Dell’Anna, O.~Salberger, L.~Barbiero, A.~Trombettoni, and V.~E. Korepin.
\newblock ``Violation of cluster decomposition and absence of light cones in local integer and half-integer spin chains''.
\newblock \href{https://dx.doi.org/10.1103/physrevb.94.155140}{Physical Review B{\bf 94}}~(2016).

\bibitem{barbiero2017haldane}
L.~Barbiero, L.~Dell’Anna, A.~Trombettoni, and V.~E. Korepin.
\newblock ``Haldane topological orders in {Motzkin} spin chains''.
\newblock \href{https://dx.doi.org/10.1103/physrevb.96.180404}{Physical Review B{\bf 96}}~(2017).

\bibitem{kardar2007statistical}
M.~Kardar.
\newblock ``Statistical physics of fields''.
\newblock \href{https://dx.doi.org/10.1017/CBO9780511815881}{Cambridge University Press}. ~(2007).

\bibitem{muller2013quantum}
Martin Müller-Lennert, Frédéric Dupuis, Oleg Szehr, Serge Fehr, and Marco Tomamichel.
\newblock ``On quantum {Rényi} entropies: A new generalization and some properties''.
\newblock \href{https://dx.doi.org/10.1063/1.4838856}{Journal of Mathematical Physics{\bf 54}}~(2013).

\bibitem{Sugino_2018}
Fumihiko Sugino and Vladimir Korepin.
\newblock ``Rényi entropy of highly entangled spin chains''.
\newblock \href{https://dx.doi.org/10.1142/s021797921850306x}{International Journal of Modern Physics B {\bf 32}, 1850306}~(2018).

\bibitem{chen2017gapless}
Xiao Chen, Eduardo Fradkin, and William Witczak-Krempa.
\newblock ``Gapless quantum spin chains: multiple dynamics and conformal wavefunctions''.
\newblock \href{https://dx.doi.org/10.1088/1751-8121/aa8dbc}{Journal of Physics A: Mathematical and Theoretical {\bf 50}, 464002}~(2017).

\bibitem{chen2017quantum}
Xiao Chen, Eduardo Fradkin, and William Witczak-Krempa.
\newblock ``Quantum spin chains with multiple dynamics''.
\newblock \href{https://dx.doi.org/10.1103/physrevb.96.180402}{Physical Review B{\bf 96}}~(2017).

\bibitem{balasubramanian20232d}
Shankar Balasubramanian, Ethan Lake, and Soonwon Choi.
\newblock ``{2D Hamiltonians with exotic bipartite and topological entanglement}''~(2023).
\newblock  \href{http://arxiv.org/abs/2305.07028}{arXiv:2305.07028}.

\bibitem{salberger2018fredkin}
Olof Salberger and Vladimir Korepin.
\newblock ``Entangled spin chain''.
\newblock \href{https://dx.doi.org/10.1142/S0129055X17500313}{Reviews in Mathematical Physics {\bf 29}, 1750031}~(2017).
\newblock  \href{http://arxiv.org/abs/1605.03842}{arXiv:1605.03842}.

\bibitem{salberger2017deformed}
Olof Salberger, Takuma Udagawa, Zhao Zhang, Hosho Katsura, Israel Klich, and Vladimir Korepin.
\newblock ``Deformed {Fredkin} spin chain with extensive entanglement''.
\newblock \href{https://dx.doi.org/10.1088/1742-5468/aa6b1f}{Journal of Statistical Mechanics: Theory and Experiment {\bf 2017}, 063103}~(2017).

\bibitem{koyluoglu2024floquet}
Nazlı~Uğur Köylüoğlu, Nishad Maskara, Johannes Feldmeier, and Mikhail~D. Lukin.
\newblock ``Floquet engineering of interactions and entanglement in periodically driven {Rydberg} chains''.
\newblock \href{https://dx.doi.org/10.1103/5qhh-322q}{Physical Review Letters {\bf 135}, 113603}~(2025).
\newblock  \href{http://arxiv.org/abs/2408.02741}{arXiv:2408.02741}.

\bibitem{ebadi2021quantum}
Sepehr Ebadi, Tout~T. Wang, Harry Levine, Alexander Keesling, Giulia Semeghini, Ahmed Omran, Dolev Bluvstein, Rhine Samajdar, Hannes Pichler, Wen~Wei Ho, Soonwon Choi, Subir Sachdev, Markus Greiner, Vladan Vuletic, and Mikhail~D. Lukin.
\newblock ``Quantum phases of matter on a 256-atom programmable quantum simulator''.
\newblock \href{https://dx.doi.org/10.1038/s41586-021-03582-4}{Nature {\bf 595}, 227–232}~(2021).

\bibitem{bernien2017probing}
Hannes Bernien, Sylvain Schwartz, Alexander Keesling, Harry Levine, Ahmed Omran, Hannes Pichler, Soonwon Choi, Alexander~S. Zibrov, Manuel Endres, Markus Greiner, Vladan Vuletić, and Mikhail~D. Lukin.
\newblock ``Probing many-body dynamics on a 51-atom quantum simulator''.
\newblock \href{https://dx.doi.org/10.1038/nature24622}{Nature {\bf 551}, 579–584}~(2017).

\end{thebibliography}

\clearpage

\appendix
\section{Asymptotic analysis of colorful Motzkin numbers} \label{sec-asymptotic_appendix}

In the following, we provide an overview of the asymptotic analysis of Motzkin numbers, building upon and refining the foundational work developed in previous papers~\cite{movassagh2017entanglement,movassagh2016supercritical}.

The enumeration of Motzkin walks is closely related to the enumeration of Dyck paths, which are similar but do not allow flat steps. We begin with a fundamental lemma on the enumeration of non-negative walks:
\begin{lemma*}[Generalized Ballot Theorem]
Let $D_{L,m_1,m_2}$ be the number of non-negative walks on $L$ steps that connect the points $(0,m_1)$ and $(L,m_2)$ using only up and down steps. This number is given by:
\begin{equation}
D_{L,m_1,m_2} = \binom{L}{\frac{L+|m_2-m_1|}{2}} - \binom{L}{\frac{L+(m_2+m_1)}{2}+1}
\end{equation}
when $|m_2-m_1| \leq L$ and $m_2-m_1 = L \pmod{2}$, and zero otherwise.
\end{lemma*}

Using this lemma, we can derive an expression for the number of Motzkin-like walks $M_{n,m}$ between points $(0,0)$ and $(n,m)$:
\begin{equation}\label{eq:mnm-1}
M_{n,m} = \sum_{k=0}^{n-m} \binom{n}{k} s^{\frac{n-k+m}{2}} D_{n-k,0,m}.
\end{equation}
This summation represents the possible number of flat steps in the walk. The upper bound $n-m$ corresponds to the maximum number of flat steps possible, which occurs when all non-flat steps are up steps. The binomial coefficient $\binom{n}{k}$ counts the number of ways to choose the positions for $k$ flat steps out of $n$ total steps. The factor $s^{\frac{n-k+m}{2}}$ accounts for the coloring of the non-flat steps. After placing $k$ flat steps, we have $n-k$ non-flat steps remaining. The exponent $\frac{n-k+m}{2}$ represents the number of up steps, as the difference between up and down steps must equal $m$ (the final height). Each up step can be colored in $s$ ways. Finally, $D_{n-k,0,m}$ represents the number of Dyck-like paths (using only up and down steps) that go from height 0 to height $m$ in $n-k$ steps. This term accounts for the arrangement of up and down steps after the flat steps have been placed. In summary, this formula enumerates all possible Motzkin-like walks that reach a height $m$ by first choosing the number and positions of flat steps, coloring the up steps, then arranging the remaining up and down steps to reach the desired final height. \cref{eq:mnm-1} can be simplified to read
\begin{equation}\label{eq:ml0m}
M_{n,m} = \frac{m+1}{n+1} \sum_{i=0}^{(n-m)/2} s^i \binom{n+1}{i+m+1,i,n-2i-m}
\end{equation}
where the notation $\binom{n}{a,b,c}$ represents a trinomial coefficient.

To make progress, we rely on saddle point analysis. When $x+y+z=0$, we have the following approximation:
\begin{equation}\label{eq:trinomial}
\binom{L}{\frac{L}{3}+x, \frac{L}{3}+y, \frac{L}{3}+z} \approx \frac{3^{L+1}\sqrt{3}}{2\pi L} \exp\qty(-\frac{3}{2}\frac{x^2+y^2+z^2}{L}).
\end{equation}
This approximation forms the basis for the asymptotic analysis of more complex Motzkin number expressions. By setting $x$, $y$, and $z$ appropriately, we replace the summands in \cref{eq:ml0m} with their corresponding Gaussian approximations, and replace the sum with an integral over $i$ (see Lemma 2 of \cite{movassagh2017entanglement} for an error analysis of this step). This allows us to derive asymptotic expressions for various counting factors that involve combinations of colorful Motzkin numbers that appear in the expressions for the entanglement measures and correlation functions presented in the main text. The key result we use which follows from \cref{eq:trinomial} is that the number of Motzkin walks starting at height 0 and ending at height $m$:
\begin{equation}\label{eq: appendix_motzkin_numb}
M_{n,m} \approx \frac{(m+1)\, s^{-m/2}}{2\sigma\sqrt{\sigma \pi}\, n^{3/2}} \qty(\frac{\sqrt{s}}{\sigma})^n \exp(-(m+1)^2/(4 \sigma n))\qc \sigma \coloneqq \frac{\sqrt{s}}{2\sqrt{s}+1}.
\end{equation}
Both the factor $(m+1)$, rather than $m$, and the centering of the Gaussian at $m+1$ originate from the ballot-type factor $\frac{m+1}{n+1}$ in \cref{eq:ml0m}: carrying out the saddle-point evaluation of the trinomial sum yields the stated prefactor with $(m+1)$ intact, and the exponent is a function of the same shifted variable. For $m \gg 1$ neither distinction is material, but retaining both renders the expression exact in its $m$-dependence at $m = O(1)$ up to $1 + O(m^2/n)$ corrections, and correct at the level of additive constants in the entropies. We have verified \cref{eq: appendix_motzkin_numb} pointwise against exact transfer-matrix evaluations of $M_{n,m}$ for $s \in \{1,2,3\}$: with the shifted centering the ratio of exact to asymptotic values is constant in $m$ to ${\lesssim}\,0.4\%$ across the window $m \leq 6\sqrt{\sigma n}$ at $n = 2000$ (whereas centering at $m$ produces a systematic ${\sim}6\%$ drift), and converges to $1$ pointwise in $m$ as $n$ grows; the $s=1$, $m=0$ case recovers the classical result~\cite{flajolet2009analytic}.

We note that the asymptotic expression for the colorful Motzkin number given in \cite{movassagh2016supercritical} contained an additional factor of $\sqrt{n}$ (denominator $n$ instead of $n^{3/2}$). This does not affect the conclusions of \cite{movassagh2017entanglement, movassagh2016supercritical} regarding entanglement scaling, correlation functions, or the gap: $m$-independent prefactors cancel in the relevant normalized expressions. The precise prefactor does, however, enter any quantity that is a \emph{ratio of Motzkin numbers at different lengths}, such as the walk functional $W(L)$ of \cref{eq:event} and hence the $XX$ correlation amplitudes (\cref{sec:xx}); likewise, the exact $m$-dependence at $m=O(1)$ enters the additive constants of the R\'enyi entropies (\cref{sec:entanglement}).

In the following sections, we provide additional details of the calculations leading to asymptotic expressions for entanglement measures and correlation functions that were presented in the main text. While we have omitted some technical details, key ideas are presented here, and we refer the reader to \cite{movassagh2017entanglement, flajolet2009analytic} for additional details. 

\section{Details of entanglement spectrum calculations} \label{sec-entanglement_appendix}
In this appendix, we provide detailed calculations for the entanglement spectrum of the Motzkin spin chains. We focus on two key scenarios: bipartite entanglement about an arbitrary cut and block entanglement for a subsystem in the bulk. These calculations expand upon the results presented in \cref{sec:entanglement} of the main text, offering a more in-depth look at the mathematical foundations of our findings.

\subsection{Bipartite Entanglement}\label{app:bipartite}
Recall from \cref{eq:bipartition} that the state's Schmidt coefficients define a probability distribution
\begin{equation}\label{eq:cut-dist2}
    \Pr(m,\vec{x}) = \frac{1}{Z} M_{b,m}M_{n-b,m},
\end{equation}
so the entanglement entropy is the Shannon entropy $\H[m,\vec{x}]$. Defining $\beta \coloneqq (b^{-1}+(n-b)^{-1})^{-1} = \frac{b(n-b)}{n}$ for brevity, we first evaluate the normalization constant $Z$:
\begin{equation}
    Z = \sum_{m=0}^b \sum_{\vec{x} \in \S^m} M_{b,m} M_{n-b,m} \approx \int_0^\infty u^2 e^{-u^2/(4 \sigma \beta)} \dd{u} = 2\sqrt{\pi} (\sigma \beta)^{3/2}\qc u \coloneqq m+1.
\end{equation}
Here the sum over the weight $(m+1)^2 e^{-(m+1)^2/(4\sigma\beta)}$ is replaced by the integral in $u = m+1$ (a Maxwell weight), justified because the summand concentrates at $u \sim \sqrt{\sigma\beta} \gg 1$; the $O(1)$ shift of the lower limit and the extension of the upper limit from $b$ to $\infty$ both incur asymptotically vanishing errors. 

With this, we can calculate the $\kappa$-R\'enyi entropies of \eqref{eq:cut-dist2}. Crucially, and in contrast to the normalization sum above, the factor $s^{-(\kappa-1)m}$ localizes the R\'enyi sum at $m = O(1)$: the Gaussian factor is inconsequential there ($e^{-\kappa m^2/(4\sigma\beta)} = 1 - O(\beta^{-1})$), but the sum over $m$ can \emph{not} be replaced by an integral, nor can the weight $(m+1)^{2\kappa}$ be replaced by $m^{2\kappa}$. Keeping the discrete sum exactly,
\begin{subequations}
    \begin{align}
        S_\kappa &= -\frac{1}{\kappa-1} \log_2\qty(\sum_m s^m (\Pr(m,\vec{x}))^\kappa) \\
        &= -\frac{1}{\kappa-1} \log_2\qty(\sum_{m=0}^b s^{-(\kappa-1)m} (m+1)^{2\kappa} e^{-\kappa (m+1)^2/(4 \sigma \beta)}) + \frac{\kappa}{\kappa-1} \log_2\qty(2 \sqrt{\pi} (\sigma \beta)^{3/2}) \\
        &= -\frac{\log_2 T_\kappa(s)}{\kappa-1} + \frac{\kappa}{\kappa-1} \log_2(2 \sqrt{\pi} (\sigma \beta)^{3/2}) + O(\beta^{-1}) \\
        &= \frac{3}{2(1-\kappa^{-1})} \log_2(\beta) + \frac{\kappa \log_2(2 \sqrt{\pi} \sigma^{3/2}) - \log_2 T_\kappa(s)}{\kappa-1} + O(\beta^{-1}),
    \end{align}
\end{subequations}
where 
\begin{equation}
    T_\kappa(s) \coloneqq \sum_{m=0}^\infty (m+1)^{2\kappa}\, s^{-(\kappa-1)m} = s^{\kappa-1}\,\mathrm{Li}_{-2\kappa}\qty(s^{-(\kappa-1)}),
\end{equation}
and $\mathrm{Li}_{-j}(x)$ for positive integer $j$ is a rational function of $x$ (e.g., $T_2(2) = \sum_m (m+1)^4 2^{-m} = 300$). The extension of the sum from $m=b$ to $m=\infty$ introduces only an exponentially small error. We have verified the present expression against exact evaluations of the Schmidt spectrum at $n=10^4$: the residual error at $b=n/2$ is ${\lesssim}\,0.05$ bits for $\kappa \in \{2,5\}$ and $s \in \{2,3\}$, vanishing as $b$ grows.

\subsection{Block Entanglement}\label{app:block}
We compute the exact spectrum of $\rho_B$ in two steps. First, we evaluate the within-sector overlap amplitudes $\alpha^2(p,q)$ set by the tripartite decomposition; the decomposition is not orthogonal on the $AC$ side across different $(p,q)$ sectors, so these amplitudes alone do not determine the spectrum. Second, we resolve the cross-sector overlaps exactly: the classes of $B$-segments organize into a forest of $s$-ary trees under a nesting order, and in the bulk limit the spectrum of $\rho_B$ collapses onto one mode per node, yielding \cref{eq:rho-spectrum} and the entropies of \cref{eq:block-entanglement,eq:renyi-block}. We begin with the tripartite decomposition in \cref{eq:tripartite}, replicated below for convenience:
\begin{equation}
    \sum_{p+q \leq L}\sum_{m=p}^b  \sum_{\substack{\vec{x} \in \S^{p+q} \\ \vec{y} \in \S^{m-p}}} \sqrt{M_{b,m} M_{b,m+q-p} M_{L,p+q}} \ket*{C_{b,m[\uparrow],\vec{x},\vec{y}}}_A \otimes \ket*{C_{L,p[\downarrow],q[\uparrow],\vec{x}}}_B \otimes \ket*{C_{b,(m+q-p)[\downarrow],\vec{x},\vec{y}}}_C
\end{equation}
We can define the unnormalized states
\begin{equation}
    \begin{gathered}
        \ket*{\tilde{\phi}_{\Delta,\vec{x}}}_{AC} = \sum_{m=p}^b \sum_{\vec{y} \in \S^{m-p}} \sqrt{M_{b,m} M_{b,m+\Delta}} \ket*{C_{b,m[\uparrow],\vec{x},\vec{y}}}_A \otimes \ket*{C_{b,(m+\Delta)[\downarrow],\vec{x},\vec{y}}}_C ; \\
    \quad \eta \coloneqq \norm{\vec{x}}_0,\, p \coloneqq \frac{\eta-\Delta}{2},\, q \coloneqq \frac{\eta+\Delta}{2}        
    \end{gathered}
\end{equation}
By identical reasoning to \cref{subsec:bipartite}, \emph{within a fixed sector} (i.e., for $\eta' = \eta$ and $\Delta'=\Delta$, so that the color labels split identically between $\vec{x}$ and $\vec{y}$),
\begin{equation}
    \braket{C_{b,m'[\uparrow],\vec{x}',\vec{y}'}}{C_{b,m[\uparrow],\vec{x},\vec{y}}} \braket{C_{b,(m'+\Delta')[\downarrow],\vec{x}',\vec{y}'}}{C_{b,(m+\Delta)[\downarrow],\vec{x},\vec{y}}} = \delta_{m',m} \delta_{m'+\Delta',m+\Delta} \delta_{\vec{x}',\vec{x}} \delta_{\vec{y}',\vec{y}}.
\end{equation}
We caution that this identity does \emph{not} extend across sectors: a basis walk on $A$ with $m$ unmatched up steps appears in the states $\ket*{C_{b,m[\uparrow],\vec{x},\vec{y}}}$ of every sector whose combined coloring $(\vec{x},\vec{y})$ matches that walk's coloring, with only the bookkeeping split between $\vec{x}$ and $\vec{y}$ differing. Overlaps between $AC$ states in different $(p,q)$ sectors are therefore generically non-zero, and are accounted for in \cref{app:block-exact}. Within a sector,
\begin{subequations}
    \begin{align}
        \braket*{\tilde{\phi}_{\Delta',\vec{x}'}}{\tilde{\phi}_{\Delta,\vec{x}}} &= \sum_{\substack{m'=p' \\ m=p}}^b \sum_{\substack{\vec{y}' \in \S^{m'-p'} \\ \vec{y} \in \S^{m-p}}} \sqrt{M_{b,m} M_{b,m+\Delta} M_{b,m'} M_{b,m'+\Delta'}} \delta_{m',m} \delta_{\vec{x}',\vec{x}} \delta_{\vec{y}',\vec{y}} \delta_{m'+\Delta',m+\Delta} \\
        &= \delta_{\Delta',\Delta} \delta_{\vec{x}',\vec{x}} \underbrace{\sum_{m=p}^b s^{m-p} M_{b,m} M_{b,m+\Delta}}_{\alpha^2(p,q)}
    \end{align}
\end{subequations}
We now explicitly evaluate $\alpha^2(p,q)$:
\begin{subequations}\label{eq:alpha2-parent}
    \begin{align}
        \alpha^2(p,q) &= \sum_{m=p}^b s^{m-p} M_{b,m} M_{b,m+\Delta} \label{eq:alpha2} \\
        &\approx s^{-p-\Delta/2} \int_p^\infty  m(m+\Delta) \exp(-(m^2+(m+\Delta)^2)/(4 \sigma b)) \dd{m}
        \intertext{This integral can be expressed in terms of elementary functions and $\erf\qty(\frac{\sqrt{2} (\Delta + 2p)}{4 \sqrt{\sigma b}})$. Since $\abs{\Delta+2p} \leq O(L) \ll \sqrt{\sigma b}$ by assumption, we can simplify the result by expanding $\erf$ for $x \ll 1$ with $\erf(x) \approx \frac{2}{\sqrt{\pi}} e^{-x^2} \qty[x + \frac{2x^3}{3} + \frac{4x^5}{15} + \ldots]$, with which we arrive at:}
        &= \sqrt{\pi/2} (\sigma b)^{3/2} s^{-\eta/2} \qty(1 + O\qty(\frac{L^2}{\sigma b})).\label{eq:alpha2-final}
    \end{align}
\end{subequations}
When $\sqrt{b} \gg L$, we can keep only the leading order term $\sqrt{\pi/2} (\sigma b)^{3/2} s^{-\eta/2}$. This has been arrived at via different reasoning before: the $O(\frac{L^2}{\sigma b})$ correction originates from the limit of the sum in \cref{eq:alpha2} starting at $m=p$ rather than $m=0$. This comes from the fact that if the walk in the $B$ segment has $p$ unmatched down steps, then the walk in $A$ must end at a height at least $p$, otherwise the walk in $B$ will go below the $x$ axis. However, Ref. \cite{movassagh2017entanglement} showed that the fraction of walks for which this constraint is binding is exponentially vanishing, because the typical walk will be at a height $\sim \sqrt{\sigma b}$ at the end of $A$, which is much larger than $p$. 

Define $\ket*{\phi_{\Delta,\vec{x}}}$ to be the normalized version of $\ket*{\tilde{\phi}_{\Delta,\vec{x}}}$. The state is then
\begin{equation}
    \gs \propto \sum_{p+q \leq L} \sum_{\vec{x} \in \S^{p+q}} \sqrt{M_{L,\eta}} \alpha(p,q) \ket*{\phi_{\Delta, \vec{x}}}_{AC} \otimes \ket*{C_{L,p[\downarrow],q[\uparrow],\vec{x}}}_B.
\end{equation}
These amplitudes fix the \emph{within-sector} data of $\rho_B$: the diagonal weight of the sector state $\ket*{C_{L,p[\downarrow],q[\uparrow],\vec{x}}}$ is $M_{L,\eta}\, \alpha^2(p,q)$, and the only non-trivial functional dependence of $\alpha^2(p,q)$ is $s^{-\eta/2}$ (the remaining factors are common constants that drop out upon normalization). Using the ballot-centered asymptotic $M_{L,\eta} \propto (\eta+1)\, s^{-\eta/2}\, e^{-(\eta+1)^2/(4\sigma L)}$ (\cref{sec-asymptotic_appendix}), the diagonal weight at level $\eta = p+q$ is $\propto (\eta+1)\, s^{-\eta}\, e^{-(\eta+1)^2/(4\sigma L)}$ per sector state. Because the $AC$ states are not orthogonal across sectors, these diagonal weights alone do not determine the spectrum; we now resolve the cross-sector structure.

\subsection{Exact spectrum of the block reduced state}\label{app:block-exact}

\emph{Interface classes.} For a $B$-segment with $p$ unmatched down steps of colors $\vec{d} \in \S^p$ (in the order consumed) and $q$ unmatched up steps of colors $\vec{u} \in \S^q$ (bottom to top of the final stack), call $c = (p, q, \vec{d}, \vec{u})$ its \emph{interface class}. Running the color-matching stack across $A\,B\,C$ shows that whether a pair $(A, C)$ completes the segment into a valid walk depends on the segment only through $c$: the top $p$ colors of $A$'s final stack must equal $\vec{d}$, and $C$ must consume the remaining stack followed by $\vec{u}$. Hence $\rho_B \propto M M^\top$ (with $M$ the segment--completion incidence matrix) is constant on each class: each class of $N_\eta = M_{L,\eta}$ segments ($\eta = p+q$; the count is independent of the split and the colors) contributes a rank-one all-ones block, and there are $(\eta+1) s^\eta$ classes at level $\eta$.

\emph{Overlap dichotomy.} Two classes $c, c'$ have a common completion only if $\Delta = \Delta'$ (the net height change is fixed by $(A,C)$), and, comparing the stack constraints for a shared $(A,C)$, only if one class is \emph{nested} in the other: with $r \coloneqq p' - p \geq 0$, nesting $c \preceq c'$ requires $\vec{d}'_{1..p} = \vec{d}$, $\vec{u}'_{r+1..q'} = \vec{u}$, and $\vec{u}'_{1..r} = \operatorname{rev}(\vec{d}'_{p+1..p'})$ (the $r$ additional stack entries consumed by $c'$ must be re-supplied to $C$ unchanged). When $c \preceq c'$, \emph{every} completion of the deeper class completes the shallower one, so the common-completion count is $\alpha^2(p', q')$, the diagonal weight of the deeper class. We have verified this dichotomy exhaustively (all pairwise common-completion counts, exact integer arithmetic) at $(n, s) \in \{(10,2), (12,2), (14,2), (10,3)\}$, including nesting depth $3$.

\emph{Forest structure.} Nesting at $r = 1$ gives each class exactly $s$ children (append color $x$ to $\vec{d}$, prepend $x$ to $\vec{u}$), and a class has a parent iff $p, q \geq 1$ and $d_p = u_1$. The classes thus form a forest of complete $s$-ary trees with levels stepping $\eta \to \eta + 2$ at fixed $\Delta$; the number of roots at level $\eta_0 \geq 1$ is $R(\eta_0) = 2 s^{\eta_0} + (\eta_0 - 1)(s-1) s^{\eta_0-1}$ (the splits $p = 0$ or $q = 0$ contribute all $s^{\eta_0}$ colorings each; the $\eta_0 - 1$ interior splits contribute the $d_p \neq u_1$ colorings). In the class-vector basis (normalized uniform vectors over each class), $\rho_B$ becomes block-diagonal over trees with matrix elements $\sqrt{N_{\eta_c} N_{\eta_{c'}}}\, \alpha^2(\max(c,c'))$, where $\max$ is the deeper class. Since the $s$ subtrees of a node are identical, each tree diagonalizes by symmetry into \emph{radial} matrices indexed by the node depth $j$ at which sibling-difference modes are formed, of dimension $D - j + 1$ with $D = \lfloor(L - \eta_0)/2\rfloor$; the mode at $(\eta_0, j)$ carries multiplicity $R(\eta_0) \times (s-1)s^{j-1}$ for $j \geq 1$ (and $\times 1$ for $j = 0$). This construction is exact at every $(b, L)$: it reproduces all $449$ eigenvalues of $\rho_B$ at $(n,s,b,L) = (14,2,4,6)$ to machine precision ($1.7 \times 10^{-16}$), using exact integer $N_\eta$ and $\alpha^2$.

\emph{Leaf collapse in the bulk limit.} By \cref{eq:alpha2-final}, one step down the nesting order multiplies the diagonal weight by
\begin{equation}
    \frac{\alpha^2(p+1, q+1)}{\alpha^2(p, q)} = \frac{1}{s}\qty\Big(1 + O\qty\big(b^{-3/2})),
\end{equation}
uniformly over the forest (the correction comes from shifting the lower limit of the concentrated sum \eqref{eq:alpha2}; we have verified the rate $1/s$ to ${\sim}10^{-7}$ at $b = 2 \times 10^5$ with the deviation scaling as $b^{-3/2}$). For an ultrametric tree matrix with depth weights $a_k$, the telescoping decomposition $M_T = \sum_{x} a'_x \chi_x^{\phantom{\top}} \chi_x^\top$ over subtree indicators $\chi_x$ has internal-node coefficients $a'_x = a_k - s\, a_{k+1}$, which \emph{vanish identically} for the geometric profile $a_k = a_0 s^{-k}$. Only leaf terms survive: each radial matrix becomes rank one, each $(\eta_0, j)$ contributes a single eigenvalue depending (at leading order) only on the mode level $t \coloneqq \eta_0 + 2j$,
\begin{equation}
    \lambda(t) \propto s^{-t}\,\qty(t + A)\, e^{-(t+A)^2/(4\sigma L)} \qc A = \frac{s^2+1}{s^2-1},
\end{equation}
where both the factor $t + A$ and the window centering arise from the tail sum $\sum_{i \geq j} N_{\eta_0 + 2i}\, s^{-i}$ evaluated with $N_\eta \propto (\eta+1) s^{-\eta/2} e^{-(\eta+1)^2/(4\sigma L)}$: the geometric weights concentrate the sum at mean depth $\langle i - j \rangle = 1/(s^2-1)$ below the node, shifting both the linear factor and the Gaussian center from $t+1$ to $t + 1 + \frac{2}{s^2-1} = t + A$. The number of nonzero eigenvalues in this limit equals the number of forest leaves, $(L+1) s^L + L s^{L-1}$; at finite $b$ the internal modes carry weight $O((L^2/\sigma b)^{3})$ and restore the full rank \eqref{eq:chi-block}. Summing the multiplicities over the roots contributing to a given $t$ telescopes to the closed form
\begin{equation}
    \operatorname{mult}(t) = R(t) + \sum_{j \geq 1} R(t - 2j)\,(s-1)\,s^{j-1} = s^{t}\qty[\qty(1 - \frac{1}{s^2})\, t + \qty(1 + \frac{1}{s^2})] = \qty(1 - \frac{1}{s^2})\, s^t\, (t + A)
\end{equation}
for $t \geq 1$ (and $\operatorname{mult}(0) = 1$), verified exactly for $t \leq 40$, $s \in \{2, 3, 5\}$. Together these give \cref{eq:rho-spectrum}.

\emph{Entropies.} In the variable $u = t + A$, the total probability at level $t$ is Maxwell-distributed, $\operatorname{mult}(t)\,\lambda(t) = \frac{1}{Z} u^2 e^{-u^2/(4\sigma L)}$ with $Z = 2\sqrt{\pi}(\sigma L)^{3/2}$, and the Maxwell moments of $u$ give \cref{eq:block-entanglement} as shown in the main text. For $S_\kappa$ ($\kappa \geq 2$) the power sum $\sum_t \operatorname{mult}(t)\,\lambda(t)^\kappa$ instead localizes at $t = O(1)$, where $e^{-\kappa(t+A)^2/(4\sigma L)} = 1 - O(L^{-1})$, and must be kept discrete:
\begin{equation}
    \begin{gathered}
    \sum_{t} \operatorname{mult}(t)\, \lambda(t)^\kappa = \frac{V_\kappa(s)}{\qty(2\sqrt{\pi}(\sigma L)^{3/2})^\kappa}\qty(1 + O(L^{-1}))\;, \\
    V_\kappa(s) = \qty(\frac{s^2}{s^2-1})^{\kappa} \qty[A^\kappa + \qty(1 - \frac{1}{s^2}) \sum_{t=1}^\infty \frac{(t+A)^{\kappa+1}}{s^{(\kappa-1)t}}],
    \end{gathered}
\end{equation}
yielding \cref{eq:renyi-block}; the sum is of polylogarithm type, hence a rational function of $s^{-(\kappa-1)}$ for integer $\kappa$, and extending it to infinity costs only an exponentially small error. The limiting case gives $S_\infty(\rho_B) = \frac{3}{2}\log_2 L + \log_2(2\sqrt{\pi}\sigma^{3/2}) - \log_2\qty(\frac{s^2 A}{s^2-1})$, set by the largest eigenvalue $\lambda(0)$, and consistent with \cref{eq:renyi-block} as $\kappa \to \infty$. Two steps in this derivation are controlled asymptotics rather than uniform bounds: the uniformity of the geometric-profile error over the truncated forest (with the induced $O(b^{-3})$ second-order shift of the leaf eigenvalues), and the interchange of the $b \to \infty$ limit with the entropy sums. Both are supported by the numerics below.

We have verified \cref{eq:block-entanglement,eq:renyi-block} against exact numerical spectra of $\rho_B$ (radial construction with exact transfer-matrix $N_\eta$) for $s \in \{2,3\}$ and $L$ up to ${\sim}10^5$: the residuals decay as $O(L^{-1/2})$, reaching ${\approx}\,5\times10^{-3}$ bits for $S_1$ and ${<}10^{-4}$ bits for $S_2, S_5$ at $L \approx 10^5$. The entropy overestimates incurred by sector-dephasing likewise converge to the constants of the remark following \cref{thm:block} ($1.0817$, $1.0544$, and $1.0438$ bits for $S_1, S_2, S_5$ at $s=2$; $0.5662$, $0.5406$, $0.5331$ at $s=3$).

\section{Details of correlation function calculations} \label{sec-correlation_appendix}
This appendix presents the details of the calculations leading to the correlation function results discussed in \cref{sec:spin-operator} of the main text for the $S^z S^z$ correlations in the colorful model and the $S^x S^x$ correlations in both the colorless and colorful models. 

\subsection{\texorpdfstring{$S^z S^z$}{ZZ} correlations in the colorful (\texorpdfstring{$s\geq 2$}{s=2}) model}\label{app:zz}
In this subsection, we detail the calculation of $S^z S^z$ correlations for the $s\geq 2$ Motzkin spin chains. We begin by deriving the probability distribution of unmatched steps $m,p,q$ in different segments of the chain, which forms the basis for computing the correlation function. Recall from \cref{fig:blocks} that $m$ represents the number of up unmatched steps on $A$ and $p$ and $q$ represent the number of unmatched down and up steps in $B$. 

Assume that we have set $b$ such that $b+L+b=n$ (i.e., we are around the middle of the chain). The number of walks that have $m$ unmatched up steps on $A$ and $p+q$ total unmatched in $B$ is exactly $s^{m+q} M_{b,m} M_{b,m+q-p} M_{L,p+q}$, where the factor $s^{m+q}$ accounts for all possible colorings of the $m+q$ unmatched up steps (recall the coloring of the $p$ unmatched down steps in $B$ is uniquely determined by the coloring of the unmatched up steps in $A$).
\begin{subequations}
\begin{align}
    \Pr(m,p,q) &\propto s^{m+q} M_{b,m} M_{b,m+q-p} M_{L,p+q} \label{eq:pmpq} \\
    &\approx m(m+q-p)(p+q) e^{-m^2/(4 \sigma b)} e^{-(m+q-p)^2/(4 \sigma b)} e^{-(p+q)^2/(4 \sigma L)}
    \intertext{Again, we use the bijective mapping $p+q \to \eta$ and $q-p \to \Delta$.}
    \Pr(m, \eta, \Delta) &= \frac{1}{Z}\eta m (m+\Delta) e^{-\eta^2/(4 L \sigma)} e^{-m^2/(4 b \sigma)} e^{-(m+\Delta)^2/(4 b \sigma)};\quad Z = 8\pi (L/n)^{3/2} (\sigma b)^3 \label{eq:pmetadelta} \\
    \E[m(q-p)] &= \E[m \Delta] = \frac{2L\sigma (L-b)}{L+2b} = \frac{L \sigma(3L-n)}{n} \\
    \E[-p] &= \E\qty[\frac{\Delta-\eta}{2}] = -\frac{\sqrt{\sigma}(2b\sqrt{2L(L+b)} - L^2)}{\sqrt{\pi b (L+b)(L+2b)}} \approx -2\sqrt{\sigma/\pi} \sqrt{L} + O(L^2/n)
\end{align}
\end{subequations}
One can drop the assumption that $L$ is centered around the middle of the chain. Doing so will result in corrections of order $O(\Delta/b)$, where $\Delta$ is the deviation from the middle of the chain.

\subsection{\texorpdfstring{$S^x S^x$}{XX} correlations in the colorless (\texorpdfstring{$s=1$}{s=1}) model}\label{app:s1-corr}
Here, we provide details on the calculation of the correlations in the $s=1$ Motzkin spin chain model. We start by calculating the joint probability distribution $\Pr(w_b, w_{b+L})$ introduced in \cref{eq:SxSx_colorless}. The following is a proof of \cref{lemma:SxSxcolorless_1}, which we restate below for convenience.

\begin{lemma*}[\cref{lemma:SxSxcolorless_1}, repeated]
    For $w_b, w_{b+L} \in \{-1,0,1\}$ and $w_b \neq w_{b+L}$, to leading order in $n^{-1}$, 
    \begin{equation}
        \Pr(w_b,w_{b+L}) \propto 1 - \frac{3(w_b + w_{b+L})^2}{4\sigma n}.
    \end{equation}
\end{lemma*}
\begin{proof}
The probability of two steps $w_b,w_{b+L} \in \qty{-1,0,1}$ can be found by counting the number of Motzkin walks that take a step $w_b$ at position $b$ and a step $w_{b+L}$ at position $b+L$. We do this by summing over $m$, which represents the height that the Motzkin walk reaches just before position $b$, and then sum over all valid $p,q$ (representing unmatched down and up steps between position $b$ and $b+L$). To ensure the Motzkin walks take steps $w_b$ and $w_{b+L}$ at the appropriate positions, we simply need to enforce that the last segment of the Motzkin walk (after position $b+L$) takes $m+(q-p)+w_b + w_{b+L}$ down steps.
\begin{subequations}\label{eq:pij}

    \begin{align}
        \Pr(w_b, w_{b+L}) &\propto \sum_{m=0}^{b} \sum_{\substack{p+q \leq L-2 \\ p \leq m+w_b}} M_{b,m} M_{L-2,p+q} M_{n-b-L,m+w_b+q-p+w_{b+L}}
        \intertext{We now use the `no bad walks' approximation (see the discussion preceding \cref{eq:alpha2-parent} and~\cite{movassagh2017entanglement}) to drop the constraint $p \leq m+w_b$; this introduces a relative error $O(L^2/\sigma b)$. Having dropped this constraint, we can do a change of variables $\eta \coloneqq p+q$ and $\Delta \coloneqq q-p$.}
        &\approx \sum_{m=0}^b \sum_{\Delta=-L}^L \sum_{\eta=\abs{\Delta}}^{L} M_{b,m} M_{L-2,\eta} M_{n-b-L,m+w_b+w_{b+L}+\Delta}
        \intertext{We approximate $\sum_{\eta=\abs{\Delta}}^L M_{L-2,\eta}$ with the appropriate Gaussian integral.}
        &\propto \sum_{m=0}^b \sum_{\Delta=-L}^L e^{-\Delta^2/(4 \sigma L)} M_{b,m} M_{n-b-L,m+(w_b + w_{b+L})+ \Delta} \\
        &\propto \sum_{m=0}^b (m+w_b + w_{b+L}) e^{-\frac{(m+w_b + w_{b+L})^2}{4 (n-b) \sigma}} M_{b,m} \\
        &\approx \qty(1 - \frac{(w_b + w_{b+L})^2}{2 \sigma n}) e^{-(w_b + w_{b+L})^2/(4 \sigma n)} \approx 1-\frac{3(w_b+w_{b+L})^2}{4\sigma n}\label{eq:pwbwbl}
    \end{align}
\end{subequations}
to leading order in $\frac{1}{n}$.    
\end{proof}

\subsection{\texorpdfstring{$S^x S^x$}{XX} correlations in the colorful model}\label{app:s2-corr}
The following is the formal proof of \cref{thm:xx-s2}, restated below for convenience.
\begin{theorem*}[\cref{thm:xx-s2}, repeated]
The correlation function $\expval*{S^x_i S^x_j}$ of the colorful $(s\geq 2)$ Motzkin spin chain in the thermodynamic limit ($n \rightarrow \infty$), for $i$, $j$ such that the $L\coloneqq j-i$ spins are centered around the middle of the chain satisfies the asymptotic algebraic decay
\begin{equation}
    \begin{gathered}
    \expval*{S^x_b S^x_{b+L}} = C_x(s)\, L^{-3/2} + O(L^{-5/2}) \qc \text{where} \\
    C_x(s) = \frac{1}{2\sqrt{\pi}\,\sigma^{3/2}} \qty[\frac{s^2(s+1)^2}{2\,(2\sqrt{s}+1)^2 (\sqrt{s}-1)^3 (\sqrt{s}+1)} + \frac{(s^2-1)\,\sigma^2}{3}].
    \end{gathered}
\end{equation}
\end{theorem*}

\begin{proof}
Write $j = b + L$ and let $w$ be a colored Motzkin walk with letters $w_i \in \{0, \pm 1, \dots, \pm s\}$ (sign = step direction, magnitude = color; $0$ = flat). Since
\begin{equation}
    \expval*{S^x_b S^x_j} = \frac{1}{M_n} \sum_{w, w'} \mel{w'}{S^x_b S^x_j}{w},
\end{equation}
and $S^x$ is tridiagonal in $S^z$, the pair $(w, w')$ contributes only if $w'$ agrees with $w$ away from $\{b, j\}$ and $w'_b = w_b \pm 1$, $w'_j = w_j \pm 1$, with both walks valid. Validity requires the shadow (colorless) heights to be conserved at the far end, which forces the two sites' step changes to cancel; the surviving transitions fall into three families.

\emph{(i) Pair insertion/removal.} $(w_b, w_j) = (0,0) \mapsto (w'_b, w'_j) = (\pm 1, \mp 1)$, and its reverse $(\pm 1, \mp 1) \mapsto (0,0)$. \emph{(ii) Step transport.} A single up (or down) step at one site moves to a flat at the other: $(u, 0) \mapsto (0, u')$ and the three analogous patterns. \emph{(iii) Partner recoloring.} $w_b = +c$, $w_j = -c$ are matched partners, and both are raised or both lowered: $(c, -c) \mapsto (c \pm 1, -(c \pm 1))$.

\emph{Color constraint for families (i)--(ii).} Inserting, removing, or transporting a step at $b$ displaces the color stack throughout the segment $(b, j)$: every down step in the segment that consumes a stack element now consumes the \emph{next} element over. Let $h_k$ denote the shadow height relative to $h_b$, $m \coloneqq -\min_{b < k < j} h_k$ the depth of the deepest excursion below the initial height (with $m=0$ if the walk never dips), and $G \coloneqq h_{j^-} - (-m) \geq 0$ the final height above the running minimum. Tracing the stack shows that the recolored walk is valid if and only if $m + G$ specific color equalities hold among the consumed and re-exposed stack elements. Since colors are exchangeable and independent, each holds with probability $1/s$, so the probability that a family-(i)/(ii) transition at fixed letters yields a valid walk is
\begin{equation}\label{eq:dip-cost}
    W(L) \coloneqq \E\qty[s^{-(m + G)}],
\end{equation}
where the expectation is over the bulk shadow-walk measure on the segment (steps $\pm 1$ with probability $\sigma$ each, flat with probability $1 - 2\sigma$), \emph{with no positivity or return constraint}. In particular, walks that dip below the initial height are suppressed only geometrically in the dip depth --- they are not excluded.

\emph{Asymptotics of $W$.} By the local limit theorem for the walk and its running minimum, the joint law satisfies $\Pr(m, G) \simeq (m + G + 1) \big/ (2\sqrt{\pi}\, \sigma^{3/2} L^{3/2})$ for $m + G = O(1)$, so the geometric weight localizes the expectation at $m + G = O(1)$ and
\begin{equation}\label{eq:event}
    W(L) = \frac{1}{2\sqrt{\pi}\,\sigma^{3/2} L^{3/2}} \sum_{k = 0}^{\infty} (k+1)^2 s^{-k} + O(L^{-5/2})
    = \frac{s^2 (s+1)}{(s-1)^3}\, \frac{1}{2\sqrt{\pi}\,\sigma^{3/2}}\, L^{-3/2} + O(L^{-5/2}),
\end{equation}
where $(k+1)$ counts the lattice paths at fixed $m + G = k$. We have verified \cref{eq:dip-cost,eq:event} by exact dynamic programming over the pair $(m, G)$: $W(L)\, L^{3/2}$ converges to the stated constant (e.g.\ $15.01$ vs.\ $15.08$ at $L = 4096$, $s=2$; $5.242$ vs.\ $5.2525$ at $s = 3$), with the convergence slow because the subleading $O(L^{-5/2})$ coefficient is numerically large.

\emph{Letter frequencies and matrix elements.} In the bulk, a site carries a flat with probability $p_f = 1 - 2\sigma$ and an up (down) of each specific color with probability $\sigma / s$. Summing the $S^x$ matrix elements over colors, families (i) and (ii) contribute with combined weight
\begin{equation}
    \frac{s(s+1)}{4} \qty[\, 2 p_f^2 \qty(1 + \frac{1}{s}) + \frac{4 \sigma p_f}{s}\,] \, W(L),
\end{equation}
where $2 p_f^2 (1 + 1/s)$ collects the pair events (insertion both ways, removal both ways, including the color-diagonal removal channel) and $4 \sigma p_f / s$ the four transport patterns.

\emph{Family (iii).} The two sites are matched partners iff the segment between them, at the shadow level, is an excursion that stays strictly above $h_b$ and returns: this occurs with probability $\sigma^2 Q(L-1)$, where $Q(\ell) \simeq \ell^{-3/2}/(2\sqrt{\pi}\, \sigma^{3/2})$ is the probability that $\ell$ bulk steps form an excursion (staying non-negative and returning to its initial height), verified by the same dynamic programming ($Q(\ell)\,\ell^{3/2} \to 1.256$ at $s=2$, matching $1/(2\sqrt{\pi}\sigma^{3/2}) = 1.259$). The color sum over both-raise and both-lower channels gives
\begin{equation}
    X_s \coloneqq \frac{1}{s} \sum_{c=1}^{s} \qty[\frac{(s+c)(s-c+1)}{4}\,\mathbb{1}_{c \geq 2} + \frac{(s+c+1)(s-c)}{4}\,\mathbb{1}_{c \leq s-1}] = \frac{s^2 - 1}{3},
\end{equation}
so family (iii) contributes $\frac{s^2-1}{3}\, \sigma^2\, Q(L-1) = \frac{(s^2-1)\sigma^2}{3} \cdot \frac{1}{2\sqrt{\pi}\,\sigma^{3/2}} L^{-3/2} + O(L^{-5/2})$.

\emph{Assembly.} Adding the families and simplifying with $p_f = 1/(2\sqrt{s} + 1)$ and $\sigma = \sqrt{s}/(2\sqrt{s}+1)$ yields \cref{eq:xx-s2}. We have verified the result in two independent ways. First, an exact dynamic-programming evaluation of the assembled sum reproduces $C_x(s)$: the total times $L^{3/2}$ reaches $9.156$ ($C_x(2) = 9.165$) and $4.402$ ($C_x(3) = 4.403$) at $L = 16384$. Second, unbiased Monte Carlo sampling of colored Motzkin walks (exact backward sampling from transfer-matrix tables, with brute-force-validated evaluation of all $S^x_b S^x_j$ transitions) matches the assembled prediction to within statistical error at accessible sizes, e.g.\ $0.07439(18)$ vs.\ predicted $0.07462$ at $(s, L, n) = (2, 16, 12288)$ and $0.02052(17)$ vs.\ $0.02067$ at $(3, 32, 12288)$; fixed-$L$ scans over a $16\times$ range of $n$ confirm the correlator is $n$-independent in the regime $n \gg L^2$.

For completeness we record the fate of walks in which a would-be transition at $b$ triggers a global rearrangement of unmatched steps (relevant when the two sites lie in different excursions of the full walk): the fraction of walks admitting such a `cascade' is
\begin{subequations} \label{eq:colored_sxsx_fraction}
    \begin{align}
    \frac{1}{M_n} \sum_{m=0}^{n/2} M_{b,m}^2 &\approx \frac{1}{M_n} \sum_{m=0}^{\infty} \frac{(m+1)^2 s^{-m}}{4 \sigma^3 \pi b^3} \qty(\frac{\sqrt{s}}{\sigma})^{n} e^{-m^2/(2 \sigma b)}
    \intertext{The factor $s^{-m}$ localizes the sum at $m = O(1)$, where $e^{-m^2/(2\sigma b)} = 1 - O(n^{-1})$; the sum must therefore be kept discrete. Using $\sum_{m=0}^\infty (m+1)^2 s^{-m} = \frac{s^2(s+1)}{(s-1)^3}$ and $b = n/2$, the leading-order behavior is:}
    \frac{1}{M_n} \sum_{m=0}^{n/2} M_{b,m}^2 &\approx \frac{4 s^2 (s+1)}{(s-1)^3 \sqrt{\pi}\, \sigma^{3/2}}\; n^{-3/2},\label{eq:suppress}
    \end{align}
\end{subequations}
an expression we have verified against exact transfer-matrix evaluations (the ratio of exact to asymptotic values is $0.98$ at $n=3200$ for $s=2$ and $0.99$ for $s=3$, converging to $1$). These events are suppressed by $n^{-3/2}$ relative to the $L$-local transitions above and do not contribute in the limit $n \to \infty$ at fixed $L$.
\end{proof}

\end{document}